\documentclass[11pt,a4paper]{article}

\usepackage{graphicx}%
\usepackage{multirow}%
\usepackage{amsmath,amssymb,amsfonts}%
\usepackage{amsthm}%
\usepackage{mathrsfs}%
\usepackage[title]{appendix}%
\usepackage{xcolor}%
\usepackage{textcomp}%
\usepackage{manyfoot}%
\usepackage{booktabs}%
\usepackage{algorithm}%
\usepackage{algorithmicx}%
\usepackage{algpseudocode}%
\usepackage{listings}%
\usepackage[table]{xcolor}
\usepackage{tabularx}
\usepackage{booktabs}
\usepackage{float}
\usepackage{doi,cite}

\newtheorem{theorem}{Theorem}[section]
\newtheorem{lemma}[theorem]{Lemma}
\newtheorem{definition}[theorem]{Definition}
\newtheorem{proposition}[theorem]{Proposition}
\newtheorem{corollary}[theorem]{Corollary}
\newtheorem{example}[theorem]{Example}

\title{Critical sets of Latin squares based on autoparatopisms}

\author{
Manuel Gonz\'alez-Regadera\thanks{Department of Applied Mathematics I, Universidad de Sevilla, Spain. \texttt{mgonzalez34@us.es}}
\and
Ra\'ul M. Falc\'on\thanks{Department of Applied Mathematics I, Universidad de Sevilla, Spain. \texttt{rafalgan@us.es}}
\and
Mar\'\i a Dolores Frau \thanks{Department of Applied Mathematics I, Universidad de Sevilla, Spain. \texttt{mdfrau@us.es}}
}

\date{}

\begin{document}

\maketitle

\begin{abstract}
In cryptography, critical sets of Latin squares have particularly been implemented to design secret sharing schemes. A main problem in these cryptographic protocols arises from absent holders of pieces of information that are common to different critical sets, because they become indispensable to recover the secret. This paper solves this problem by making use of the orbits of entries described by the autoparatopism group of the Latin square under consideration. To this end, we introduce the more general problem of computing critical sets of Latin squares having a given paratopism in their autoparatopism group. These critical sets depend only on the conjugacy class of the autoparatopism and the main class of the Latin square under consideration. Based on this fact, as an illustrative example, we determine the smallest and largest sizes of critical sets associated with autoparatopisms of Latin squares of order up to six. We implement this approach in the design of a new secret sharing scheme.
\end{abstract}



\maketitle

\section{Introduction}\label{sec1}

A {\em partial Latin square} of order $n$ is an $n\times n$ partial array with nonempty entries in a given set of $n$ symbols such that each symbol appears at most once per row and at most once per column. This is a {\em Latin square} if there are no empty entries in the array.

In cryptography, Latin squares and partial Latin squares have turned out to be useful in the design of new authentication schemes \cite{Denes1992}, block ciphers \cite{Stinson2007}, cipher systems \cite{Kong2008}, cryptosystems \cite{Desoky2024}, cryptographic primitives \cite{Mileva2010}, encryption-decryption algorithms \cite{Zolfaghari2022}, error-correcting codes \cite{Gupta2017}, image encryption \cite{Shen2022}, or pseudo-random sequences \cite{FalconAlvarez19}, among others. Of particular interest for the purposes of this paper, we highlight their implementation to design new secret sharing schemes \cite{Cooper1994,Fitina2006,Chum2010,Donovan2012}.

These schemes were independently introduced by Shamir \cite{Shamir1979} and Blakley \cite{Blakley1979}. In this cryptographic protocol, a trusted dealer divides the key into pieces of information that are distributed among a group of participants. Any authorized subgroup that shares at least a predetermined number of pieces of information is sufficient to recover the key, while no subgroup that shares less than the threshold value can do so. Cooper, Donovan, and Seberry described \cite{Cooper1994} a secret sharing scheme in which the key is a given Latin square $L$, whose order is made public. Then, its entries are distributed to a group of participants so that, if $L$ is the unique Latin square containing all the entries shared by a subgroup of participants, then they recover the key. Otherwise, they cannot do so.  If the removal of any participant from this group makes this recovery impossible, then their entries constitute a {\em critical set} of $L$, a notion that plays a relevant role in this cryptographic protocol. In practice, the huge growth of the number of Latin squares when the order increases makes any brute-force attack unfeasible for discovering the key unless the attacker knows a wide amount of pieces of information. The problem of absent participants is addressed by considering multiple critical sets of different sizes that necessarily have one or more entries in common. The owners of these entries have a much higher hierarchy in the scheme, because they could even become indispensable to recover the key. This could be a handicap if they were disabled. This paper solves this problem by avoiding the need to search for different critical sets with common entries. More specifically, our approach makes use of the symmetries of the Latin square under consideration to get substitutes for these absent participants. 

To formally introduce our approach, let $\mathcal{PL}(n)$ and $\mathcal{L}(n)$ denote, respectively, the set of partial Latin squares of order $n$ with nonempty entries in the set $[n]:=\{1,\ldots,n\}$, and its subset of Latin squares. Every partial Latin square $P:=(P[i,j])\in\mathcal{PL}(n)$ is uniquely identified with its set of nonempty entries $\mathrm{Ent}(P):=\{(i,j,P[i,j])\colon\, i,j,P[i,j]\in [n]\}$. The cardinality of this set is the {\em size} of $P$. Furthermore, for each pair $P_1,P_2\in \mathcal{PL}(n)$, it is said that $P_1$ {\em contains} $P_2$ if $\mathrm{Ent}(P_2)\subseteq \mathrm{Ent}(P_1)$. By abuse of notation, we denote this fact by $P_2\subseteq P_1$ (or by $P_2\subset P_1$ if $\mathrm{Ent}(P_2)\subset \mathrm{Ent}(P_1)$). If $P_1\in\mathcal{L}(n)$, then $P_2$ is {\em completable} to $P_1$, and $P_1$ is a {\em completion} of $P_2$. If $P_2$ is not completable to any other Latin square, then it is {\em uniquely completable}. If no partial Latin square $P_3\subset P_2$ is also uniquely completable to $P_1$, then $P_2$ is a {\em critical set} of $P_1$. This is {\em minimal} if there does not exist any other critical set of $P_1$ of smaller size. 

Both problems of deciding the existence of a completion and deciding whether a partial Latin square is uniquely completable are NP-complete \cite{Colbourn84,Colbourn1984}. A greedy algorithm to find critical sets in Latin squares was described in \cite{Hamalainen2007}. We refer to \cite{Keedwell2004} and \cite{Cavenagh2008} for a pair of comprehensive surveys on critical sets of Latin squares. 

Let $L\in\mathcal{L}(n)$. It is known \cite{Donovan2000,Donovan2000a} the existence of critical sets of size $\mathrm{cs}(L)$ with
{\small \[\left\lfloor \frac {n^2}4\right\rfloor\leq \mathrm{cs}(L)\leq \frac{n^2-n}2.\]}
The smallest and largest critical set of $L$ are respectively denoted by $\mathrm{scs}(L)$ and $\mathrm{lcs}(L)$. Those ones of any given Latin square of order $n$ are respectively denoted by $\mathrm{scs}(n)$ and $\mathrm{lcs}(n)$. Table \ref{table:scs_lcs} shows the known values of $\mathrm{scs}(n)$ and $\mathrm{lcs}(n)$ (see \cite{Curran1979,Donovan1995,Bate1999,Adams2001,Adams2003,Bean2005}).

\begin{table}[h]
\caption{Smallest and largest sizes of critical sets}\label{table:scs_lcs}%
\renewcommand{\tabcolsep}{7pt}
\begin{tabular}{@{}lllllllll@{}}
\toprule
$n$ & 1  & 2 & 3 & 4 & 5 & 6 & 7 & 8\\
\midrule
$\mathrm{scs}(n)$   & 0   & 1  & 2 & 4 & 6 & 9 & 12 & 16 \\
$\mathrm{lcs}(n)$     & 0   & 1 & 3 & 7 & 11 & 18  \\
\hline
\end{tabular}
\end{table}

Some partial results concerning the exact value of $\mathrm{scs}(n)$ are also known for some families of Latin squares such as back-circulant \cite{Howse1998,Bate2003} or symmetric ones \cite{Mojdeh2007}. More generally, it is known \cite{Smetaniuk1979,Curran1979, Cavenagh2007,Bean2000,Bean2003,Hatami2003} that $n\left\lfloor \frac 12 (\log n)^{1/3}\right\rfloor \leq \mathrm{scs}(n)\leq \left\lfloor \frac {n^2}4\right\rfloor$ and $n^2\left(1-\frac{2+\ln 2}{\ln n}\right)+n\left(1+\frac{\ln(8\pi)}{\ln n}\right)-\frac {\ln 2}{\ln n}\leq \mathrm{lcs}(n)\leq n^2-3n+3$. Both upper bounds are known to be reached for some values (see Table \ref{table:scs_lcs}). That of $\mathrm{scs}(n)$ is conjectured to be always exact. 

The number and sizes of critical sets of a given Latin square depend only on its paratopism class \cite{Donovan1995}. Recall here that, if $S_n$ denotes the symmetric group on the set $[n]$, then a {\em paratopism} is any pair $\Theta:=(\pi;f)$ formed by a permutation $\pi\in S_3$ and a triple $f:=(f_1,f_2,f_3)\in S_n\times S_n\times S_n$. Then, a partial Latin square $P\in \mathcal{PL}(n)$ is said to be {\em paratopic} to $P^\Theta\in \mathcal{PL}(n)$ if

{\footnotesize\begin{equation}\label{eq_paratopism}
\mathrm{Ent}(P^\Theta)=\left\{\left(f_{\pi(1)}(e_{\pi(1)}),\,f_{\pi(2)}(e_{\pi(2)}),\,f_{\pi(3)}(e_{\pi(3)})\right)\colon\, \left(e_1,\,e_2,\, e_3\right)\in\mathrm{Ent}(P)\right\}.
\end{equation}}
The semidirect product $S_3\wr S_n:=S_3\ltimes (S_n\times S_n\times S_n)$ is a group under the composition of paratopisms
\begin{equation}\label{equation_composition}
(\rho;g)(\pi;f):=\left(\pi\rho;g^\pi f\right),
\end{equation}
where $g^\pi:=\left(g_{\pi^{-1}(1)},g_{\pi^{-1}(2)},g_{\pi^{-1}(3)}\right)$. In particular, 

\begin{equation}\label{eq_inverse}
(\pi;f)^{-1}:=\left(\pi^{-1};(f^{-1})^{\pi^{-1}}\right).
\end{equation}

\newpage

The paratopism $(\pi;f)$ is an {\em isotopism} if $\pi$ is the trivial permutation $\mathrm{Id}_3\in S_3$; and an {\em isomorphism} if, in addition, $f_1=f_2=f_3$. Furthermore, if $f$ is trivial (that is, if its three components coincide with the trivial permutation $\mathrm{Id}_n$ on the set $[n]$), then $\Theta$ is called a {\em parastrophism} from $P$ to $P^\Theta$. If this is the case, then $P$ and $P^\Theta$ are said to be {\em parastrophic}. To be isotopic, isomorphic, parastrophic or paratopic constitute equivalence relations among partial Latin squares, which give respectively rise to the so-called {\em isotopism}, {\em isomorphism}, {\em conjugacy} and {\em main classes} of partial Latin squares. The distribution of Latin squares into these classes is known \cite{Hulpke2011, Kolesova1990, McKay2007} for order up to $11$, and that of partial Latin squares is known \cite{Falcon2013, Falcon2015, Falcon2018, FalconStones2020} for order up to six. A census of critical sets is known for every isotopism and main class of Latin squares of order up to six \cite{Adams2003}, while an example of a critical set is known for each main class of Latin squares of order seven \cite{Donovan1998}.

An {\em autoparatopism} of a partial Latin square $P\in\mathcal{PL}(n)$ (respectively, {\em autotopism}) is a paratopism (respectively, isotopism) from $P$ to itself. In this regard, autoparatopisms and autotopisms describe the symmetries of any partial Latin square. The set $\mathrm{Apar}(P)$ (respectively, $\mathrm{Atop}(P)$) formed by the autoparatopisms (respectively, autotopisms) of $P$ is a group under the composition described in (\ref{equation_composition}). In fact, every finite group is isomorphic to the autoparatopism group of at least one partial Latin square \cite{Stones2013}.

The concepts of completability and critical set of Latin squares have been extended \cite{Falcon2013} to those based on a given isotopism. More precisely, a partial Latin square $P\in\mathcal{PL}(n)$ is said to be {\em $f$-completable}, with $f\in S_n\times S_n\times S_n$, if there is a completion $L\in\mathcal{L}(n)$ of the former such that $f\in\mathrm{Atop}(L)$. If $P$ is not $f$-completable to any other Latin square, then it is {\em uniquely $f$-completable}. It is a {\em $f$-critical set} of $L$ if no partial Latin square properly contained in $P$ is also uniquely $f$-completable. If $f$ is the trivial isotopism, then the classical notions of completability and critical set arise. A census of critical sets based on isotopisms is known \cite{Falcon2021} for all $n\leq 5$.

This paper further generalizes these concepts by dealing with critical sets of Latin squares containing a given paratopism in their autoparatopism group. The paper is organized as follows. In Section~\ref{sec:preliminaries} we extend the concepts of completability and critical set from isotopisms to paratopisms. Then, we show how the orbits of entries induced by the autoparatopism under consideration constitute a first approach to construct these critical sets. In Section \ref{sec:conjugacy}, we show that the set of critical sets based on a given autoparatopism of a Latin square depends only on the conjugacy class of the former and the main class of the latter. The study of cases concerning these conjugacy classes is carried out according to the cycle structures of their respective paratopisms. This enables us to compute the sizes of critical sets based on autoparatopisms of Latin squares of order up to six. Finally, in Section \ref{sec:secret}, we implement these critical sets based on paratopisms to design a generalization of the secret sharing scheme described in \cite{Cooper1994} that is not so dependent on the existence of absent but indispensable participants.

\section{Critical sets based on autoparatopisms}\label{sec:preliminaries}

In this section, we introduce the notions of completability and critical set of a Latin square with respect to a given paratopism $\Theta\in S_3\wr S_n$ as a natural generalization of those based on isotopisms. To this end, we define the set
\[\mathcal{L}(\Theta):=\left\{L\in\mathcal{L}(n)\colon\, \Theta\in\mathrm{Apar}(L)\right\}.\]

\begin{definition}\label{definition_CriticalSet} A partial Latin square $P\in\mathcal{PL}(n)$ is {\em $\Theta$-completable} if it is completable to a Latin square $L\in\mathcal{L}(\Theta)$. If $P$ is not $\Theta$-completable to any other Latin square, then it is {\em uniquely $\Theta$-completable}. Moreover, it is a {\em $\Theta$-critical set} of $L$ if no partial Latin square properly contained in $P$ is also uniquely $\Theta$-completable. We denote by $\mathrm{CS}_\Theta(L)$ the set of $\Theta$-critical sets of $L$. The smallest and largest $\Theta$-critical set of $L$ are respectively denoted by $\mathrm{scs}_\Theta(L)$ and $\mathrm{lcs}_\Theta(L)$. Those ones of any given Latin square of order $n$ are respectively denoted by $\mathrm{scs}_\Theta(n)$ and $\mathrm{lcs}_\Theta(n)$.
\end{definition}

Classical critical sets of Latin squares arise for the trivial autoparatopism $\Theta=(\mathrm{Id}_3;(\mathrm{Id}_n,\mathrm{Id}_n,\mathrm{Id}_n))$, whereas those based on a given isotopism do for $\Theta=(\mathrm{Id}_3;(f_1,f_2,f_3))\in S_3\wr S_n$. Definition \ref{definition_CriticalSet} also includes critical sets of symmetric Latin squares, which was introduced in \cite{Mojdeh2007} for even orders. This refers to Latin squares that are equal to their own transpose, and arises for $\Theta=((12);(\mathrm{Id}_n,\mathrm{Id}_n,\mathrm{Id}_n))$. In any case, the set $\mathrm{CS}_\Theta(L)$ is preserved by $\Theta$. To see it, we first show the relationship among (uniquely) completable partial Latin squares of two paratopic Latin squares.

\begin{lemma}\label{lemma_powers} If $P\in\mathcal{PL}(n)$ is (uniquely) $\Theta$-completable to $L\in\mathcal{L}(\Theta)$, then $P^\Theta$ is also (uniquely) $\Theta$-completable to $L$.
\end{lemma}

\begin{proof} Since $P^\Theta\subseteq L^\Theta=L$, we have that $P^\Theta$ is $\Theta$-completable to $L$. Furthermore, if $P$ is uniquely $\Theta$-completable, but $P^\Theta\subseteq L'$  for some $L'\in\mathcal{L}(\Theta)\setminus\{L\}$, then $P=\left(P^\Theta\right)^{\Theta^{-1}}\subseteq {L'}^{\Theta^{-1}}=L'$, which is a contradiction. Thus, $P^\Theta$ is also uniquely $\Theta$-completable.
\end{proof}

\begin{proposition}\label{proposition_powers} Let $L\in\mathcal{L}(\Theta)$ and $P\in\mathrm{CS}_\Theta(L)$. Then, $P^\Theta\in \mathrm{CS}_\Theta(L)$.
\end{proposition}
\begin{proof} From Lemma \ref{lemma_powers}, $P^\Theta$ is uniquely $\Theta$-completable to $L$. Thus, $P^\Theta\not\in \mathrm{CS}_\Theta(L)$ if and only if it contains properly a partial Latin square $Q$ that is uniquely $\Theta$-completable to $L$. But then, $Q^{\Theta^{-1}}$ is properly contained in $P$. Moreover, it is $\Theta$-completable to $L$ from Lemma \ref{lemma_powers}. This contradicts the fact that $P\in\mathrm{CS}_\Theta(L)$. Hence, $P^\Theta\in \mathrm{CS}_\Theta(L)$.
\end{proof}

\vspace{0.2cm}

From now on, we denote $\Theta^k:=\Theta^{k-1}\Theta$ for every positive integer $k$, where $\Theta^0$ is the trivial paratopism. It is well defined because $S_3\wr S_n$ is a group under the composition described in (\ref{equation_composition}). The {\em order} of $\Theta$ is the smallest positive integer $o(\Theta)$ such that $\Theta^{o(\Theta)}=\Theta^0$. 

\begin{lemma}\label{lemma_bounds} Let $L\in\mathcal{L}(\Theta)$ and let $k\leq o(\Theta)$ be a non-negative integer. Every $\Theta^k$-critical set of $L$ contains a $\Theta$-critical set of $L$.
\end{lemma}

\begin{proof} For any Latin square $L'$, the set $\mathrm{Apar}(L')$ is a group, so $\Theta\in \mathrm{Apar}(L')$ implies $\Theta^k\in \mathrm{Apar}(L')$. Thus, $L(\Theta)\subseteq L(\Theta^k)$. As a consequence, if $P\in\mathrm{CS}_{\Theta^k}(L)$, then the set of $\Theta$-completions of $P$ is a subset of the set of $\Theta^k$-completions of $P$. Since $P$ is uniquely $\Theta^k$-completable to $L$, we have that $P$ is also uniquely $\Theta$-completable to $L$. Since $P$ is finite, the non-empty collection of subsets of $P$ that are uniquely $\Theta$-completable to $L$ has a minimal element $Q$ with respect to inclusion. By Definition \ref{definition_CriticalSet}, $Q$ is a $\Theta$-critical set of $L$, and $Q\subseteq P$.
\end{proof}

\vspace{0.2cm}

As an immediate consequence, the values described in Table \ref{table:scs_lcs} are upper bounds of $\mathrm{scs}_\Theta(L)$. In particular, this value is zero if $n=1$, and $1$ if $n=2$. The same two trivial cases hold for $\mathrm{lcs}_\Theta(L)$, which can be checked directly. More generally, the next result holds readily from Lemma \ref{lemma_bounds}.

\begin{proposition}\label{proposition_bounds} If $L\in\mathcal{L}(\Theta)$, then
\[\mathrm{scs}_\Theta(L)\leq \min_{0\leq k< o(\Theta)}\left\{\mathrm{scs}_{\Theta^k}(L)\right\}\hspace{1cm}\text{and}\hspace{1cm}\mathrm{scs}_\Theta(n)\leq \min_{0\leq k< o(\Theta)}\left\{\mathrm{scs}_{\Theta^k}(n)\right\}.\]
\end{proposition}

A first approach to understand much better all the notions described in Definition \ref{definition_CriticalSet} consists of making use of the orbits of entries that are induced by the paratopism under consideration. Note in this regard from (\ref{eq_paratopism}) that every autoparatopism $\Theta:=(\pi;(f_1,f_2,f_3))\in S_3\wr S_n$ of a Latin square $L\in\mathcal{L}(n)$ acts faithfully on the set of entries $\mathrm{Ent}(L)$ so that
\[\begin{array}{cccc}
\Theta: & \mathrm{Ent}(L) & \to & \mathrm{Ent}(L)\\
& e:=(e_1,e_2,e_3) & \mapsto &\Theta(e):=\left(f_{\pi(1)}(e_{\pi(1)}),\,f_{\pi(2)}(e_{\pi(2)}),\,f_{\pi(3)}(e_{\pi(3)})\right).
\end{array}\]
The following definition generalizes in a natural way that one described in \cite{Falcon2021} about the orbits of entries induced by a given autotopism. The projection onto the first two coordinates of each entry gives rise, in fact, to the notion of cell orbit induced by a given autoparatopism, which was described in \cite{Mendis2017} (or by a given autotopism, in \cite{Stones2012}). 

\begin{definition}\label{definition_orbit} Let $L\in\mathcal{L}(\Theta)$. The {\em $\Theta$-orbit} of an entry $e\in\mathrm{Ent}(L)$ is the set 
\[\mathrm{Orb}_{\Theta}(e):=\left\{\Theta^k(e)\colon\, 0\leq  k\leq o(\Theta)\right\}\subseteq \mathrm{Ent}(L).\]
We denote by $\mathrm{Orb}_{\Theta}(L)$ the set formed by all $\Theta$-orbits in $\mathrm{Ent}(L)$.
\end{definition}

The next result illustrates the relevant role that the set $\mathrm{Orb}_{\Theta}(L)$ plays in the study of $\mathrm{CS}_\Theta(L)$. It holds readily from the fact that every $\Theta$-orbit in $L$ is uniquely determined by any of its entries. 

\begin{lemma}\label{lemma_orbit} Every $\Theta$-critical set of a Latin square $L\in\mathcal{L}(\Theta)$ has at most one entry in each $\Theta$-orbit of $L$. As a consequence,
\begin{equation}\label{eq_orb}
\mathrm{scs}_{\Theta}(L)\leq \mathrm{lcs}_{\Theta}(L)\leq |\mathrm{Orb}_{\Theta}(L)|.
\end{equation}
\end{lemma}

Taking into account the previous lemma, Algorithm \ref{alg_1} describes a procedure to determine all feasible sizes of $\Theta$-critical sets of $L$.

\begin{algorithm}
\caption{Feasible sizes in $\mathrm{CS}_\Theta(L)$.}\label{alg_1}
\begin{algorithmic}[1]
\Procedure{$\text{\em Sizes}$}{$\Theta$, $L$}
\Comment{{\small Input: $\Theta\in S_3\wr S_n$ and $L\in\mathcal{L}(\Theta)$.}}
\State $K:=$\text{ Empty list}
\State $E:=\left\{e_1,\ldots,e_{|\mathrm{Orb}_{\Theta}(L)|}\right\}$ \text{ such that } $\mathrm{Orb}_\Theta(e_i)\neq \mathrm{Orb}_\Theta(e_j)$ \text{ if } $i\neq j$
\For{$k\gets 1, |\mathrm{Orb}_{\Theta}(L)|$}
    \State $E_k:=\{S\subseteq E\colon\, |S|=k\}$
\EndFor
\For{$k\gets 1, |\mathrm{Orb}_{\Theta}(L)|$}
        \For{$S\in E_k$}
            \If{$S\in \mathrm{CS}_\Theta(L)$}
                \State $K\gets K\cup\{k\}$
                \State \textbf{break}
            \EndIf
        \EndFor
\EndFor
\State \textbf{return} $K$
\EndProcedure
\end{algorithmic}
\end{algorithm}

As we have already mentioned in the introductory section, deciding unique completability of a partial Latin square is NP-complete in general \cite{Colbourn84,Colbourn1984}. However, enforcing invariance under non-trivial autoparatopisms $\Theta$ drastically reduces the search space of candidate completions $\mathcal{L}(\Theta)$. It allows us to search over $\Theta$-orbits of entries rather than individual cells. In this regard, for a Latin square $L\in\mathcal{L}(n)$, Algorithm \ref{alg_1} filters subsets of representative entries from the orbit partition $Orb_\Theta(L)$. So, the computational complexity is $2^{|\mathrm{Orb}_{\Theta}(L)|}$ instead of $2^{n^2}$.

At the end of this section, Example \ref{example_upperbound} shows that the upper bound in (\ref{eq_orb}) can indeed be reached. Furthermore, to be in the same $\Theta$-orbit is an equivalence relation among the entries of $L$, so $\mathrm{Orb}_{\Theta}(L)$ is a partition of $\mathrm{Ent}(L)$ and hence, of the cells of $L$. We visually represent this partition by coloring two cells of $L$ with the same color if and only if they correspond to the same $\Theta$-orbit. Of course, if $\Theta$ is the trivial paratopism, there are so many colors as cells. We term this representation a {\em $\Theta$-coloring} of $L$.

\begin{example}\label{example_CriticalSet} We consider the paratopism
\[\Theta:=\left((12);(\mathrm{Id}_3,(123),(12))\right)\in S_3\wr S_3.\]
From (\ref{equation_composition}), we have $o(\Theta)=6$ and the compositions
\[\Theta^2=\left(\mathrm{Id}_3;((123),(123),\mathrm{Id}_3)\right),\hspace{1cm} \Theta^3=\left((12);((123),(132),(12))\right)\]
\[\Theta^4=\left(\mathrm{Id}_3;((132),(132),\mathrm{Id}_3)\right)\hspace{0.4cm}\text{and} \hspace{0.4cm}\Theta^5=\left((12);((132),\mathrm{Id}_3,(12))\right).\]
We claim that
\[L\equiv\begin{array}{|c|c|c|} \hline
1 & 3 & 2 \\ \hline
2 & 1 & 3 \\ \hline
3 & 2 & 1 \\ \hline
\end{array}\in\mathcal{L}\left(\Theta^k\right)\]
for every positive integer $k\leq o(\Theta)$. Since $\mathrm{Apar}(L)$ is a group, it is enough to prove this fact for $k=1$. In this regard, note that $\Theta$ acts on the entries of $L$ as follows.
\[(1,1,1)\xrightarrow{\Theta} (2,1,2) \xrightarrow{\Theta} (2,2,1)\xrightarrow{\Theta} (3,2,2) \xrightarrow{\Theta} (3,3,1) \xrightarrow{\Theta} (1,3,2) \xrightarrow{\Theta} (1,1,1)\]
\[(1,2,3)\xrightarrow{\Theta} (3,1,3) \xrightarrow{\Theta} (2,3,3) \xrightarrow{\Theta} (1,2,3)\]
Based on this action, we have the following $\Theta^k$-colorings of $L$.
\[\begin{array}{ccccccc}
\begin{array}{|c|c|c|}\hline
\cellcolor{red!100}{\color{white} 1} & \cellcolor{blue!100}{\color{white} 3} & \cellcolor{red!100}{\color{white} 2}\\\hline
\cellcolor{red!100}{\color{white} 2} & \cellcolor{red!100}{\color{white} 1} & \cellcolor{blue!100}{\color{white} 3}\\ \hline
\cellcolor{blue!100}{\color{white} 3} & \cellcolor{red!100}{\color{white} 2} & \cellcolor{red!100}{\color{white} 1}\\ \hline
\end{array} & \phantom{\hspace{1cm}}  &  &
\begin{array}{|c|c|c|}\hline
\cellcolor{red!100}{\color{white} 1} & \cellcolor{blue!100}{\color{white} 3} & \cellcolor{green!100}{\color{white} 2}\\\hline
\cellcolor{green!100}{\color{white} 2} & \cellcolor{red!100}{\color{white} 1} & \cellcolor{blue!100}{\color{white} 3}\\ \hline
\cellcolor{blue!100}{\color{white} 3} & \cellcolor{green!100}{\color{white} 2} & \cellcolor{red!100}{\color{white} 1}\\ \hline
\end{array}  & \phantom{\hspace{1cm}}   &  &
\begin{array}{|c|c|c|}\hline
\cellcolor{red!100}{\color{white} 1} & \cellcolor{blue!100}{\color{white} 3} & \cellcolor{green!100}{\color{white} 2}\\\hline
\cellcolor{pink!100}{\color{white} 2} & \cellcolor{green!100}{\color{white} 1} & \cellcolor{orange!100}{\color{white} 3}\\ \hline
\cellcolor{yellow!100}{\color{white} 3} & \cellcolor{red!100}{\color{white} 2} & \cellcolor{pink!100}{\color{white} 1}\\ \hline
\end{array}\\
k\in\{1,5\} &  & & k\in \{2,4\}&  & & k=3
\end{array}\]
Now, we consider the partial Latin squares
\[P_1\equiv\begin{array}{|c|c|c|} \hline
1 & \phantom{3} & \phantom{2} \\ \hline
  &   &   \\ \hline
  &   &   \\ \hline
\end{array} \hspace{1cm} \text{and} \hspace{1cm} P_2\equiv\begin{array}{|c|c|c|} \hline
1 & 3 & \phantom{2} \\ \hline
  &   &   \\ \hline
  &   &   \\ \hline
\end{array}.\]
In addition, let $k\in\{1,3,5\}$ and $l\in\{0,2,4\}$. It is readily verified that $P_1\in\mathrm{CS}_{\Theta^k}(L)\setminus \mathrm{CS}_{\Theta^0}(L)$. Moreover, $P_2\in\mathrm{CS}_{\Theta^l}(L)\setminus \mathrm{CS}_{\Theta^k}(L)$, because $P_1\subset P_2$. In case of being interested in all the $\Theta^k$-critical sets of $L$, the following study of cases arises.
\begin{itemize}
    \item For $k\in\{1,5\}$, every partial Latin square in $\mathcal{PL}(3)$ that is formed by exactly one entry in $\mathrm{Orb}_{\Theta^k}((1,1,1))$ is a $\Theta^k$-critical set of $L$. This is not the case for any subset of entries in $\mathrm{Orb}_{\Theta^k}((1,2,3))$, because the resulting partial Latin square would also be $\Theta^k$-completable to
    
\[\begin{array}{|c|c|c|} \hline
\cellcolor{red!100}{\color{white} 2} & \cellcolor{blue!100}{\color{white} 3} & \cellcolor{red!100}{\color{white} 1}\\\hline
\cellcolor{red!100}{\color{white} 1} & \cellcolor{red!100}{\color{white} 2} & \cellcolor{blue!100}{\color{white} 3}\\ \hline
\cellcolor{blue!100}{\color{white} 3} & \cellcolor{red!100}{\color{white} 1} & \cellcolor{red!100}{\color{white} 2}\\ \hline
\end{array}\in\mathcal{L}\left(\Theta^k\right),\]
for which we have also described a $\Theta^k$-coloring. As a consequence, $\mathrm{scs}_{\Theta^k}(L)=\mathrm{lcs}_{\Theta^k}(L)=1$. 

\vspace{0.2cm}

\item For $k\in\{2,4\}$, we note that every $\Theta^k$-critical set of $L$ is formed by two entries in different orbits. Thus,  $\mathrm{scs}_{\Theta^k}(L)=\mathrm{lcs}_{\Theta^k}(L)=2$.

\vspace{0.2cm}

\item For $k=3$, we observe that every partial Latin square in $\mathcal{PL}(3)$ that is formed by exactly one entry in $\mathrm{Orb}_{\Theta^k}((1,1,1))\cup \mathrm{Orb}_{\Theta^k}((1,3,2))\cup\mathrm{Orb}_{\Theta^k}((2,1,2))$ is a $\Theta^k$-critical set of $L$. Since the symbol appearing in the remaining $\Theta^k$-orbits is mandatorily $3$, we have that $\mathrm{scs}_{\Theta^k}(L)=\mathrm{lcs}_{\Theta^k}(L)=1$.\hfill $\lhd$
\end{itemize}
\end{example}

In Example \ref{example_CriticalSet}, every $\Theta^k$-critical set, with $k\in\{1,5\}$, necessarily contains an entry in $\mathrm{Orb}_{\Theta^k}((1,1,1))$. In general, it is useful to know the existence of these types of orbits for which every critical set must contain one of its entries. We finish this section with some results in this regard. 

Previously, for each paratopism $\Theta:=(\pi;f)\in S_3\wr S_n$, with $f:=(f_1,f_2,f_3)$, and each triple $e:=(e_1,e_2,e_3)\in [n]^3$, we define the following sets.

\begin{align*}\mathrm{Fix}_{\mathrm{row}}(e):& =\{(\pi;f)\in S_3\wr S_n\colon\, f_{\pi(1)}(e_{\pi(1)})=e_1\}\\
\mathrm{Fix}_{\mathrm{col}}(e):& =\{(\pi;f)\in S_3\wr S_n\colon\, f_{\pi(2)}(e_{\pi(2)})=e_2\}\\
\mathrm{Fix}_{\mathrm{sym}}(e):& =\{(\pi;f)\in S_3\wr S_n\colon\, f_{\pi(3)}(e_{\pi(3)})=e_3\}
\end{align*}
In addition, for each Latin square $L\in\mathcal{L}(n)$, we define the subgroup
\[\mathrm{Apar}_{(12)}(L):=\mathrm{Atop}(L)\cup \left\{((12);f)\in \mathrm{Apar}(L)\right\}\leq\mathrm{Apar}(L).\]

\begin{theorem}\label{theorem_bitrade} Let $L\in\mathcal{L}(n)$ and $\Theta:=(\pi;f)\in \mathrm{Apar}_{(12)}(L)$ be such that there exist an entry $e\in\mathrm{Ent}(L)$ and a positive integer $k_0\leq o(\Theta)$ such that
\begin{enumerate}
    \item[{\rm (C1)}] $\Theta\in \mathrm{Fix}_{\mathrm{row}}(e)\setminus \mathrm{Fix}_{\mathrm{col}}(e)$;
    \item[{\rm (C2)}] $\Theta^{k_0}\in \mathrm{Fix}_{\mathrm{col}}(e)$; and
    \item[{\rm (C3)}] $\Theta^{o(\Theta)-k_0+1}\in \mathrm{Fix}_{\mathrm{sym}}(e)$.
\end{enumerate}
Then, every $\Theta$-critical set of $L$ contains an entry in $\mathrm{Orb}_\Theta(e)$. The same happens if the first two conditions are respectively replaced by 
\begin{enumerate}
    \item[{\rm (C1')}] $\Theta\in \mathrm{Fix}_{\mathrm{col}}(e)\setminus \mathrm{Fix}_{\mathrm{row}}(e)$.
    \item[{\rm (C2')}] $\Theta^{k_0}\in \mathrm{Fix}_{\mathrm{row}}(e)$.
\end{enumerate}
\end{theorem}

\begin{proof} If $P\subset L$ does not contain any entry in $\mathrm{Orb}_{\Theta}(e)$, then Theorem 3.2 in \cite{Cavenagh2025} implies that $P$ is $\Theta$-completable to the Latin square $L'\in\mathcal{L}(\Theta)$, where
\[\mathrm{Ent}(L')=\left(\mathrm{Ent}(L)\setminus\mathrm{Orb}_{\Theta}(e)\right)\cup \left\{\Theta^k((e_1,e_2,f_3(e_3)))\colon\,0\leq k<o(\Theta)\right\}.\]
Hence, $P\not\in\mathrm{CS}_\Theta(L)$, and the result holds.    
\end{proof}

\vspace{0.2cm}

Example \ref{example_CriticalSet} illustrates this theorem for $e=(1,1,1)$ and $k_0=5$. It is readily verified that $\Theta$ and $\Theta^5$ satisfy both conditions (C1') and (C2'), and $\Theta$ holds condition (C3). Theorem \ref{theorem_bitrade} also enables us to find an example in which the upper bound in Lemma \ref{lemma_orbit} is reached. 

\begin{example}\label{example_upperbound} We consider the paratopism
\[\Theta:=((12);(\mathrm{Id}_4,(1234),(12)(34)))\in S_3\wr S_4\]
and the Latin square
\[L\equiv\begin{array}{|c|c|c|c|} \hline
\cellcolor{red!100}{\color{white} 1} & \cellcolor{blue!100}{\color{white} 3} & \cellcolor{blue!100}{\color{white} 4} & \cellcolor{red!100}{\color{white} 2}\\\hline
\cellcolor{red!100}{\color{white} 2} & \cellcolor{red!100}{\color{white} 1} & \cellcolor{blue!100}{\color{white} 3} & \cellcolor{blue!100}{\color{white} 4}\\ \hline
\cellcolor{blue!100}{\color{white} 4} & \cellcolor{red!100}{\color{white} 2} & \cellcolor{red!100}{\color{white} 1} & \cellcolor{blue!100}{\color{white} 3}\\ 
\cellcolor{blue!100}{\color{white} 3} & \cellcolor{blue!100}{\color{white} 4} & \cellcolor{red!100}{\color{white} 2} & \cellcolor{red!100}{\color{white} 1}\\ 
\end{array}\in\mathcal{L}\left(\Theta^k\right)\]
for which we describe a $\Theta$-coloring. In particular, $o(\Theta)=8$ and
\[\Theta^7=((12);((1432),\mathrm{Id}_4,(12)(34)))\in S_3\wr S_4.\]
It is readily verified that $\Theta$ and $\Theta^7$ satisfy both conditions (C1') and (C2') for both entries $(1,1,1)$ and $(1,2,3)$. Since $\Theta$ also holds Condition (C3), Theorem \ref{theorem_bitrade} implies that every $\Theta$-critical set of $L$ must contain an entry in $\mathrm{Orb}_{\Theta}((1,1,1))$ and also an entry in $\mathrm{Orb}_{\Theta}((1,2,3))$. Hence, $\mathrm{scs}_\Theta(L)=\mathrm{lcs}_\Theta(L)=|\mathrm{Orb}_\Theta(L)|=2$. \hfill $\lhd$
\end{example}

\vspace{0.2cm}

The set $\mathrm{Orb}_\Theta(e)$ in Theorem \ref{theorem_bitrade} constitutes indeed the set of entries of a {\em Latin trade} of $L$. That is, a partial Latin square $P\subset L$ such that there is another partial Latin square $P'\subset L$ of the same size such that $\mathrm{Ent}(P)\cap\mathrm{Ent}(P')=\emptyset$; and $\left(\mathrm{Ent}(L)\setminus \mathrm{Ent}(P)\right)\cup \mathrm{Ent}(P')$ is the set of entries of a new Latin square. (We refer again to Theorem 3.2 in \cite{Cavenagh2025} for more details.) From now on, we denote by $\tau_\Theta(L)$ the number of distinct Latin trades in the set $\mathrm{Orb}_{\Theta}(L)$ whose entries satisfy the conditions described in Theorem \ref{theorem_bitrade}. The following corollary follows readily from that theorem.

\begin{corollary}\label{corollary_bitrade} Let $L\in\mathcal{L}(n)$ and $\Theta:=(\pi;f)\in \mathrm{Apar}_{(12)}(L)$. Then, 
\[\tau_\Theta(L)\leq \mathrm{scs}_\Theta(L).\] 
\end{corollary}

\vspace{0.1cm}

Example \ref{example_upperbound} illustrates that this lower bound is reached. In addition, the next result shows that $\tau_\Theta(L)$ is an invariant in the conjugacy class of the paratopism $\Theta$ within the subgroup $\mathrm{Apar}_{(12)}(L)\leq S_3\wr S_n$. (Recall here that two elements $a$ and $b$ in a group $G$ are {\em conjugate} if there is a third element $c\in G$ such that $b=cac^{-1}$.) 

\begin{corollary}\label{corollary_tau}  Let us consider a Latin square $L\in\mathcal{L}(n)$, and two conjugate autoparatopisms $\Theta_1:=(\pi_1;f_1)$ and $\Theta_2:=(\pi_2;f_2)$ in the group $\mathrm{Apar}_{(12)}(L)$. Then $\tau_{\Theta_1}(L)=\tau_{\Theta_2}(L)$.
\end{corollary}

\begin{proof} Since $\Theta_1$ and $\Theta_2$ are conjugate within $\mathrm{Apar}_{(12)}(L)$, there exists a third autoparatopism $\Theta_3:=(\pi_3;f_3)\in\mathrm{Apar}_{(12)}(L)$ such that $\Theta_2=\Theta_3\Theta_1\Theta_3^{-1}$. Let us suppose that conditions (C1--C3) in Theorem \ref{theorem_bitrade} hold for an entry $e\in\mathrm{Ent}(L)$, a positive integer $k_0\leq o(\Theta_1)$, and the autoparatopism $\Theta_1$. (A similar reasoning follows from conditions (C1'-C2') and (C3).) Since $\Theta_3\in\mathrm{Apar}_{(12)}(L)$, we have $\pi_3\in \{\mathrm{Id}_3,\,(12)\}$, so $\pi_3(3)=3$. Since $\Theta_2^m(\Theta_3(e))=\Theta_3(\Theta_1^m(e))$ for $m:=o(\Theta_1)-k_0+1$, and $\Theta_1^m(e)$ agrees with $e$ in the third coordinate by (C3) for $\Theta_1$, it follows that  $\Theta_2^m(\Theta_3(e))$ agree with $\Theta_3(e)$ in the third component as well. That is, Condition (C3) holds for $\Theta_2$, the entry $\Theta_3(e)\in\mathrm{Ent}(L)$ and the positive integer $k_0\leq o(\Theta_2)=o(\Theta_1)$. In addition, both conditions (C1) and (C2) hold if $\pi=\mathrm{Id}_3$, whereas both conditions (C1') and (C2') hold if $\pi=(12)$. Hence, $\tau_{\Theta_1}(L)\leq \tau_{\Theta_2}(L)$. The reciprocal follows similarly and hence, the result holds.   
\end{proof}

\vspace{0.1cm}

This last corollary is useful to determine whether two autoparatopisms in $\mathrm{Apar}_{(12)}(L)$ are conjugate. The following example illustrates this fact.

\begin{example}\label{example_conjugacy} The paratopisms $\Theta_1:=((12); ((45), (12)(3465), (12)(3465)))$ and $\Theta_2:=((12); ((12)(3465), (36), (12)(3465)))$, both of them in $S_3\wr S_6$, are autoparatopisms of the Latin square $L_{6.1}$ described in Figure \ref{Fig_MainClasses}, for which we have the following $\Theta_1$- and $\Theta_2$-colorings.

\[{\scriptsize \begin{array}{ccc}
\begin{array}{|c|c|c|c|c|c|}  \hline
\cellcolor{red!100}{\color{white} 1} & \cellcolor{red!100}{\color{white} 2} & \cellcolor{blue!100}{\color{white} 3} & \cellcolor{green!100}{\color{white} 4} & \cellcolor{pink!100}{\color{white} 5} & \cellcolor{orange!100}{\color{white} 6} \\ \hline
\cellcolor{red!100}{\color{white} 2} & \cellcolor{red!100}{\color{white} 1} & \cellcolor{pink!100}{\color{white} 4} & \cellcolor{orange!100}{\color{white} 3} & \cellcolor{blue!100}{\color{white} 6} & \cellcolor{green!100}{\color{white} 5} \\ \hline
\cellcolor{pink!100}{\color{white} 3} & \cellcolor{blue!100}{\color{white} 5} & \cellcolor{yellow!100}{\color{white} 6} & \cellcolor{brown!100}{\color{white} 1} & \cellcolor{yellow!100}{\color{white} 4} & \cellcolor{gray!100}{\color{white} 2} \\ \hline
\cellcolor{blue!100}{\color{white} 4} & \cellcolor{pink!100}{\color{white} 6} & \cellcolor{yellow!100}{\color{white} 5} & \cellcolor{gray!100}{\color{white} 2} & \cellcolor{yellow!100}{\color{white} 3} & \cellcolor{brown!100}{\color{white} 1} \\ \hline
\cellcolor{orange!100}{\color{white} 5} & \cellcolor{green!100}{\color{white} 3} & \cellcolor{gray!100}{\color{white} 1} & \cellcolor{purple!100}{\color{white} 6} & \cellcolor{brown!100}{\color{white} 2} & \cellcolor{purple!100}{\color{white} 4} \\ \hline
\cellcolor{green!100}{\color{white} 6} & \cellcolor{orange!100}{\color{white} 4} & \cellcolor{brown!100}{\color{white} 2} & \cellcolor{purple!100}{\color{white} 5} & \cellcolor{gray!100}{\color{white} 1} & \cellcolor{purple!100}{\color{white} 3} \\ \hline
\end{array} & \phantom{\hspace{2cm}} & 
\begin{array}{|c|c|c|c|c|c|}  \hline
\cellcolor{red!100}{\color{white} 1} & \cellcolor{red!100}{\color{white} 2} & \cellcolor{blue!100}{\color{white} 3} & \cellcolor{green!100}{\color{white} 4} & \cellcolor{pink!100}{\color{white} 5} & \cellcolor{orange!100}{\color{white} 6} \\ \hline
\cellcolor{red!100}{\color{white} 2} & \cellcolor{red!100}{\color{white} 1} & \cellcolor{pink!100}{\color{white} 4} & \cellcolor{orange!100}{\color{white} 3} & \cellcolor{blue!100}{\color{white} 6} & \cellcolor{green!100}{\color{white} 5} \\ \hline
\cellcolor{green!100}{\color{white} 3} & \cellcolor{orange!100}{\color{white} 5} & \cellcolor{yellow!100}{\color{white} 6} & \cellcolor{brown!100}{\color{white} 1} & \cellcolor{gray!100}{\color{white} 4} & \cellcolor{brown!100}{\color{white} 2} \\ \hline
\cellcolor{orange!100}{\color{white} 4} & \cellcolor{green!100}{\color{white} 6} & \cellcolor{gray!100}{\color{white} 5} & \cellcolor{brown!100}{\color{white} 2} & \cellcolor{yellow!100}{\color{white} 3} & \cellcolor{brown!100}{\color{white} 1} \\ \hline
\cellcolor{blue!100}{\color{white} 5} & \cellcolor{pink!100}{\color{white} 3} & \cellcolor{purple!100}{\color{white} 1} & \cellcolor{gray!100}{\color{white} 6} & \cellcolor{purple!100}{\color{white} 2} & \cellcolor{yellow!100}{\color{white} 4} \\ \hline
\cellcolor{pink!100}{\color{white} 6} & \cellcolor{blue!100}{\color{white} 4} & \cellcolor{purple!100}{\color{white} 2} & \cellcolor{yellow!100}{\color{white} 5} & \cellcolor{purple!100}{\color{white} 1} & \cellcolor{gray!100}{\color{white} 3} \\ \hline
\end{array}\\
\Theta_1\text{-coloring} & & \Theta_2\text{-coloring} 
\end{array}}\]
The only Latin trade of $L$ whose entries satisfy the conditions described in Theorem \ref{theorem_bitrade} for $\Theta_1$ is formed by the entries of $\mathrm{Orb}_{\Theta_1}((1,1,1))$. In addition, there are three Latin trades of $L$ satisfying these conditions for $\Theta_2$. They are respectively formed by the entries of $\mathrm{Orb}_{\Theta_2}((1,1,1))$, $\mathrm{Orb}_{\Theta_2}((3,4,1))$ and $\mathrm{Orb}_{\Theta_2}((5,3,1))$. Hence, Corollary \ref{corollary_tau} implies that $\Theta_1$ and $\Theta_2$ are not conjugate within $\mathrm{Apar}_{(12)}(L)$. By the way, after implementing Algorithm \ref{alg_1} in both cases, we obtain
\[\mathrm{scs}_{\Theta_1}(L_{6.1})=\mathrm{lcs}_{\Theta_1}(L_{6.1})=4 \hspace{1cm}\text{and}\hspace{1cm} \mathrm{scs}_{\Theta_2}(L_{6.1})=\mathrm{lcs}_{\Theta_2}(L_{6.1})=5.\]
We illustrate here a pair of $\Theta_1$- and $\Theta_2$-critical sets of $L$.
\[{\scriptsize \begin{array}{ccc}
\begin{array}{|c|c|c|c|c|c|} \hline
1&\phantom{2}&3&4&\phantom{5}&\phantom{6}\\\hline
&&&&&\\\hline
&&&1&&\\\hline
&&&&&\\\hline
&&&&&\\\hline
&&&&&\\\hline
\end{array} & \phantom{\hspace{2cm}} &
\begin{array}{|c|c|c|c|c|c|} \hline
1&\phantom{2}&3&4&\phantom{5}&\phantom{6}\\\hline
&&&&&\\\hline
&&&1&&\\\hline
&&&&&\\\hline
&&1&&&\\\hline
&&&&&\\\hline
\end{array}\\
\Theta_1\text{-critical set} & & \Theta_2\text{-critical set} 
\end{array}}
\]
\hfill $\lhd$
\end{example}

\section{Conjugacy classes of autoparatopisms}\label{sec:conjugacy}

In what follows, we prove that the smallest and largest sizes of critical sets based on a given autoparatopism $\Theta \in S_3\wr S_n$ of a Latin square $L\in\mathcal{L}(n)$ depend only on the conjugacy class of $\Theta$ within the subgroup $\mathrm{Apar}(L)\leq S_3\wr S_n$ and the main class of $L$. First, we show that every paratopism $\Theta_0\in S_3\wr S_n$ constitutes a one-to-one correspondence between the autoparatopism group of any partial Latin square $P\in\mathcal{PL}(n)$ and that of $P^{\Theta_0}$, which in turn preserves the conjugacy classes of their respective paratopisms.

\begin{lemma}\label{lemma_uc} If $P\in\mathcal{PL}(n)$ and $\Theta_0\in S_3\wr S_n$, then
\[\mathrm{Apar}\left(P^{\Theta_0}\right)=\left\{\Theta_0\Theta_1\Theta_0^{-1}\colon\,\Theta_1\in \mathrm{Apar}(P)\right\}.\]
\end{lemma}

\begin{proof} First, if $\Theta_1\in\mathrm{Apar}(P)$, then $\Theta_0\Theta_1\Theta_0^{-1}\in\mathrm{Apar}(P^\Theta)$, because
\[\left(P^{\Theta_0}\right)^{\Theta_0\Theta_1\Theta_0^{-1}}=\left(\left(\left(P^{\Theta_0}\right)^{\Theta_0^{-1}}\right)^{\Theta_1}\right)^{\Theta_0}=\left(P^{\Theta_1}\right)^{\Theta_0}=P^{\Theta_0}.\]
Similarly, if $\Theta_2\in \mathrm{Apar}\left(P^{\Theta_0}\right)$, then $\Theta_0^{-1}\Theta_2\Theta_0 \in \mathrm{Apar}(P)$. Thus, the result holds because $\Theta_2=\Theta_0(\Theta_0^{-1}\Theta_2\Theta_0)\Theta_0^{-1}$.
\end{proof}

\vspace{0.1cm}

Then, we prove that both values $\mathrm{scs}_\Theta(L)$ and $\mathrm{lcs}_\Theta(L)$ depend only on the conjugacy class of $\Theta$ within $\mathrm{Apar}(L)$.

\begin{proposition}\label{proposition_conjugate} Let $L\in\mathcal{L}(n)$. If two autoparatopisms $\Theta_1,\,\Theta_2\in\mathrm{Apar}(L)$ are conjugate within $\mathrm{Apar}(L)$, then there is a one-to-one correspondence between both sets $\mathrm{CS}_{\Theta_1}(L)$ and $\mathrm{CS}_{\Theta_2}(L)$ so that
\[\mathrm{scs}_{\Theta_1}(L)=\mathrm{scs}_{\Theta_2}(L)\hspace{1cm}\text{and}\hspace{1cm}\mathrm{lcs}_{\Theta_1}(L)=\mathrm{lcs}_{\Theta_2}\left(L\right).\]
\end{proposition}

\begin{proof} Since $\Theta_1$ and $\Theta_2$ are conjugate within $\mathrm{Apar}(L)$, there is a paratopism $\Theta_3\in\mathrm{Apar}(L)$ such that $\Theta_2=\Theta_3\Theta_1\Theta_3^{-1}$. Then, the result follows from Lemma \ref{lemma_uc} and Definition \ref{definition_CriticalSet}.    
\end{proof}

\vspace{0.1cm}

Example \ref{example_conjugacy} illustrates that considering in Proposition \ref{proposition_conjugate} the conjugacy within the autoparatopism group is mandatory. Being conjugate within $\mathrm{Apar}(L)$ is an equivalence relation among the autoparatopisms of $L$, which we denote by $\sim$. By Proposition \ref{proposition_conjugate}, if we want to compute all feasible sizes of critical sets based on autoparatopisms of $L$, then it suffices to implement Algorithm \ref{alg_1} on a representative paratopism of each class in the quotient group $\mathrm{Apar}(L)/\sim$. In addition, if one is interested in computing these values for all Latin squares in $\mathcal{L}(n)$, the following results show that it suffices to consider a representative Latin square of each main class.

\begin{proposition}\label{proposition_uc_orbits} If two Latin squares $L_1,L_2\in\mathcal{L}(n)$ are paratopic, then there is a one-to-one correspondence between their respective sets of critical sets based on their autoparatopisms.
\end{proposition}

\begin{proof} Since $L_1$ and $L_2$ are paratopic, there is a paratopism $\Theta_0\in S_3\wr S_n$ such that $L_1^{\Theta_0}=L_2$. Then, the result holds because Lemma \ref{lemma_uc}, together with Definition \ref{definition_CriticalSet}, implies that $P\in\mathrm{CS}_{\Theta_1}(L_1)$ for some $\Theta_1\in \mathrm{Apar}(L_1)$ if and only if $P^{\Theta_0}\in\mathrm{CS}_{\Theta_0\Theta_1\Theta_0^{-1}}(L_2)$ for some $\Theta_1\in \mathrm{Apar}(L_1)$.
\end{proof}

\begin{theorem}\label{theorem_Wanless} If $L\in\mathcal{L}(\Theta)$, then there is a paratopism $\Theta_0\in S_3\wr S_n$ such that the paratopism $\Theta_0\Theta\Theta_0^{-1}$ is either an isotopism or a paratopism of one of the following two types.
\begin{itemize}
    \item {\bf Type I}: $\left((12);(\mathrm{Id}_n,f_2,f_3)\right)\in S_3\wr S_n$.
    \item {\bf Type II}: $\left((123);(\mathrm{Id}_n,\mathrm{Id}_n,f_3)\right)\in S_3\wr S_n$.
\end{itemize}
In particular, there is a one-to-one correspondence between the sets $\mathrm{CS}_\Theta(L)$ and $\mathrm{CS}_{\Theta_0\Theta\Theta_0^{-1}}\left(L^{\Theta_0}\right)$.
\end{theorem}

\begin{proof} Theorem 2.2 in \cite{Mendis2016} implies the existence of the paratopism $\Theta_0$. Then, the mentioned correspondence follows from Proposition \ref{proposition_uc_orbits}.
\end{proof}

From now on, we denote by $\mathfrak{P}(n)$ the subset of $S_3\wr S_n$ that is formed by isotopisms and paratopisms of types I or II. Algorithm \ref{alg_2} implements both Proposition \ref{proposition_conjugate} and Theorem \ref{theorem_Wanless} to determine the sets of feasible sizes of critical sets based on the quotient group $\mathrm{Apar}(L)/\sim$ for a representative Latin square $L$ of a given main class. It uses the procedure {\em Sizes} described in Algorithm \ref{alg_1}.

\begin{algorithm}
\caption{Feasible sizes of critical sets based on $\mathrm{Apar}(L)/\sim$.}\label{alg_2}
\begin{algorithmic}[1]
\Procedure{$\text{\em MainSizes}$}{$L$}
\Comment{{\small Input: $L\in\mathcal{L}(n)$.}}
\State $A:=\mathrm{Apar}(L)/\sim$
\State $K:=\emptyset$
\For{$\Theta\in A$}
    \If{$\Theta\in \mathfrak{P}(n)$}
       \State $K\gets K\cup\{\Theta,\,\text{\em Sizes}(\Theta,L)\}$
    \Else
       \State $K\gets K\cup\left\{\Theta,\,\text{\em Sizes}\left(\Theta_0\Theta\Theta_0^{-1},L^{\Theta_0}\right)\right\}$,\, \text{ where } $\Theta_0\Theta\Theta_0^{-1}\in \mathfrak{P}(n)$
    \EndIf
\EndFor
\State \textbf{return} $K$
\EndProcedure
\end{algorithmic}
\end{algorithm}

\newpage

We have implemented Algorithm \ref{alg_2} in the {\sc GAP} system {\em (Groups, Algorithms, Programming)} \cite{GAP} to determine all feasible sizes of critical sets based on paratopisms of Latin squares of order $3\leq n\leq 6$. (The case $n<3$ is trivial.) The results are shown in Appendix \ref{secA}. In all tables therein, the column $L$ refers to the main class of the Latin square. Their representatives are shown in Figure \ref{Fig_MainClasses}. The column {\em class} refers to the conjugacy class of the autoparatopism. The columns {\em scs} and {\em lcs} show, respectively, the smallest and largest values of any critical set in each case. The column {\em time} gives the computational time in seconds required to determine all the critical sets in each case on a {\em $13^{\mathrm{th}}$ Gen Intel Core i9-13900H CPU @ 2.60GHz with 32 GB RAM}. Specific examples of critical sets associated with each case are listed in Appendix \ref{secA1}. Figure~\ref{fig:time} illustrates the global elapsed time required for the computation of critical sets based on non-trivial autoparatopisms.

\renewcommand{\tabcolsep}{2pt}
\begin{figure}[ht]
\centering
\tiny
$\begin{array}{c}
\begin{array}{ccccccccc}
\begin{tabular}{|c|c|c|}\hline
1 & 2 & 3\\\hline
2 & 3 & 1\\ \hline
3 & 1 & 2\\ \hline
\end{tabular} & \phantom{\hspace{0.5cm}} &
\begin{tabular}{|c|c|c|c|}\hline
1 & 2 & 3 & 4\\ \hline
2 & 1 & 4 & 3\\ \hline
3 & 4 & 1 & 2\\ \hline
4 & 3 & 2 & 1\\ \hline
\end{tabular}   & \phantom{\hspace{0.5cm}} & \begin{tabular}{|c|c|c|c|}\hline
1 & 2 & 3 & 4\\ \hline
2 & 1 & 4 & 3\\ \hline
3 & 4 & 2 & 1\\ \hline
4 & 3 & 1 & 2\\ \hline
\end{tabular}   & \phantom{\hspace{0.5cm}} & \begin{tabular}{|c|c|c|c|c|}\hline
1 & 2 & 3 & 4 & 5\\ \hline
2 & 3 & 4 & 5 & 1\\ \hline
3 & 4 & 5 & 1 & 2\\ \hline
4 & 5 & 1 & 2 & 3\\ \hline
5 & 1 & 2 & 3 & 4\\ \hline
\end{tabular}   & \phantom{\hspace{0.5cm}} & \begin{tabular}{|c|c|c|c|c|}\hline
1 & 2 & 3 & 4 & 5\\ \hline
2 & 1 & 4 & 5 & 3\\ \hline
3 & 4 & 5 & 1 & 2\\ \hline
4 & 5 & 2 & 3 & 1\\ \hline
5 & 3 & 1 & 2 & 4\\ \hline
\end{tabular}\\
L_3 & & L_{4.1} & & L_{4.2} & & L_{5.1} & & L_{5.2}
\end{array} \\ \\
\begin{array}{cccccc}
\begin{tabular}{|c|c|c|c|c|c|} \hline
1&2&3&4&5&6\\\hline
2&1&4&3&6&5\\\hline
3&5&6&1&4&2\\\hline
4&6&5&2&3&1\\\hline
5&3&1&6&2&4\\\hline
6&4&2&5&1&3\\\hline
\end{tabular}  &
\begin{tabular}{|c|c|c|c|c|c|}\hline
1&2&3&4&5&6\\\hline
2&3&1&5&6&4\\\hline
3&1&2&6&4&5\\\hline
4&6&5&1&3&2\\\hline
5&4&6&2&1&3\\\hline
6&5&4&3&2&1\\\hline
\end{tabular} &
\begin{tabular}{|c|c|c|c|c|c|}\hline
1&2&3&4&5&6\\\hline
2&3&1&5&6&4\\\hline
3&1&2&6&4&5\\\hline
4&6&5&3&2&1\\\hline
5&4&6&1&3&2\\\hline
6&5&4&2&1&3\\\hline
\end{tabular} &
\begin{tabular}{|c|c|c|c|c|c|} \hline
1&2&3&4&5&6\\\hline
2&3&1&6&4&5\\\hline
3&1&2&5&6&4\\\hline
4&6&5&2&1&3\\\hline
5&4&6&1&3&2\\\hline
6&5&4&3&2&1\\\hline
\end{tabular}  &
\begin{tabular}{|c|c|c|c|c|c|}\hline
1&2&3&4&5&6\\\hline
2&3&1&6&4&5\\\hline
3&4&5&2&6&1\\\hline
4&1&6&5&2&3\\\hline
5&6&2&3&1&4\\\hline
6&5&4&1&3&2\\\hline
\end{tabular} &
\begin{tabular}{|c|c|c|c|c|c|}\hline
1&2&3&4&5&6\\\hline
2&4&5&1&6&3\\\hline
3&1&2&6&4&5\\\hline
4&3&6&5&1&2\\\hline
5&6&1&2&3&4\\\hline
6&5&4&3&2&1\\\hline
\end{tabular}\\
L_{6.1} & L_{6.2} & L_{6.3} & L_{6.4} & L_{6.5} & L_{6.6}
\end{array} \\ \\
\begin{array}{cccccc}
\begin{tabular}{|c|c|c|c|c|c|} \hline
1&2&3&4&5&6\\\hline
2&4&5&1&6&3\\\hline
3&1&2&6&4&5\\\hline
4&5&6&2&3&1\\\hline
5&6&4&3&1&2\\\hline
6&3&1&5&2&4\\\hline
\end{tabular}  &
\begin{tabular}{|c|c|c|c|c|c|}\hline
1&2&3&4&5&6\\\hline
2&4&5&1&6&3\\\hline
3&5&6&2&4&1\\\hline
4&6&1&3&2&5\\\hline
5&1&4&6&3&2\\\hline
6&3&2&5&1&4\\\hline
\end{tabular} &
\begin{tabular}{|c|c|c|c|c|c|}\hline
1&2&3&4&5&6\\\hline
2&4&5&1&6&3\\\hline
3&6&4&2&1&5\\\hline
4&3&6&5&2&1\\\hline
5&1&2&6&3&4\\\hline
6&5&1&3&4&2\\\hline
\end{tabular} &
\begin{tabular}{|c|c|c|c|c|c|} \hline
1&2&3&4&5&6\\\hline
2&5&6&3&1&4\\\hline
3&6&2&1&4&5\\\hline
4&3&5&2&6&1\\\hline
5&4&1&6&3&2\\\hline
6&1&4&5&2&3\\\hline
\end{tabular}  &
\begin{tabular}{|c|c|c|c|c|c|}\hline
1&2&3&4&5&6\\\hline
2&6&4&3&1&5\\\hline
3&5&6&1&2&4\\\hline
4&3&5&2&6&1\\\hline
5&4&1&6&3&2\\\hline
6&1&2&5&4&3\\\hline
\end{tabular} &
\begin{tabular}{|c|c|c|c|c|c|}\hline
1&2&3&4&5&6\\\hline
2&6&4&5&3&1\\\hline
3&4&2&6&1&5\\\hline
4&5&1&2&6&3\\\hline
5&3&6&1&2&4\\\hline
6&1&5&3&4&2\\\hline
\end{tabular}\\
L_{6.7} & L_{6.8} & L_{6.9} & L_{6.10} & L_{6.11} & L_{6.12}
\end{array}
\end{array}$
\vspace{0.2cm}
\caption{Representative Latin squares of the main classes of $\mathcal{L}(n)$, with $3\leq n\leq 6$.}\label{Fig_MainClasses}
\end{figure}

\begin{figure}[ht]
\centering
\includegraphics[width=0.9\textwidth]{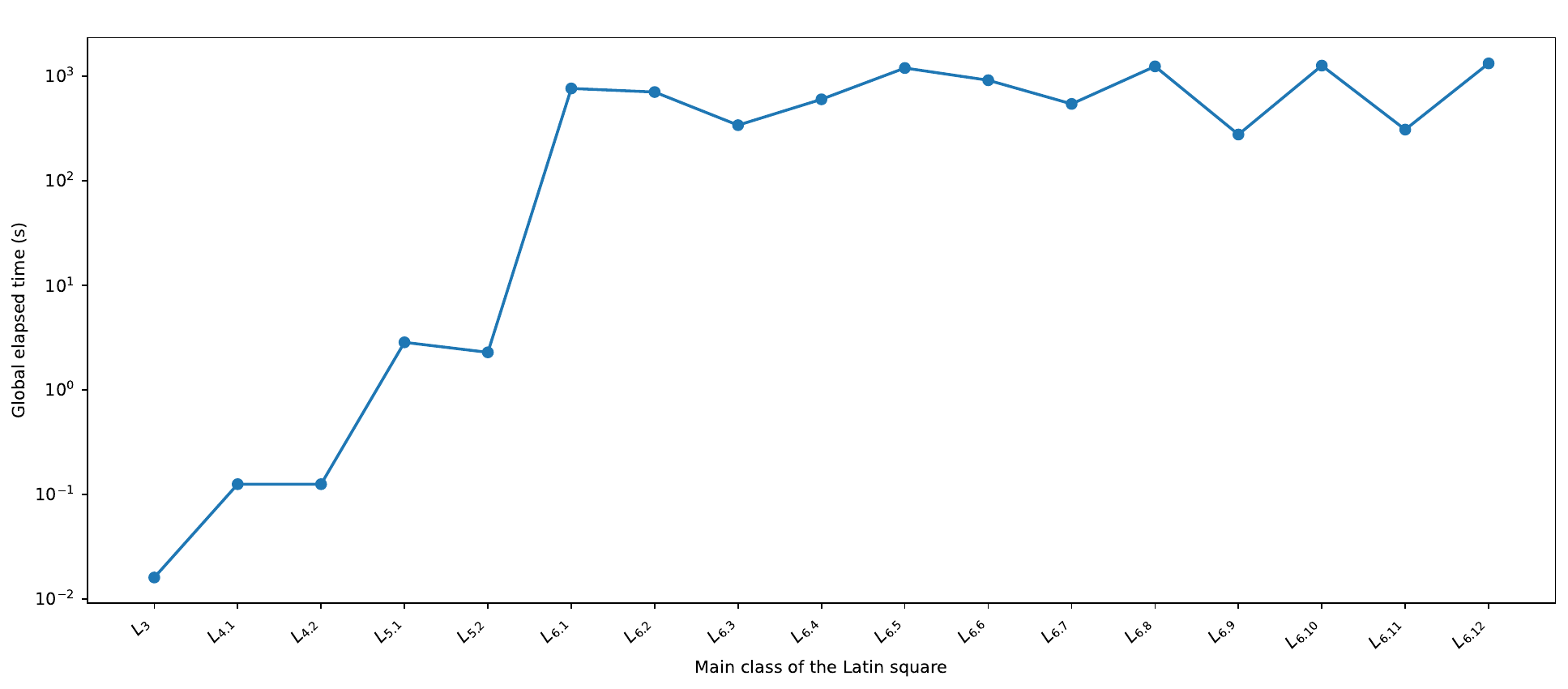}
\caption{Global computational time for non-trivial autoparatopisms.}\label{fig:time}
\end{figure}

To avoid enumeration of all paratopisms, we carry out our study according to their cycle structures, which are represented in the column $z$ of our tables. Recall here that the {\em cycle structure} of a permutation $\pi\in S_n$ is the expression $z_{\pi}:=n^{\lambda_n^{\pi}}\ldots 1^{\lambda_1^{\pi}}$, where $\lambda_\ell^\pi$ is the number of cycles of length $\ell$ in the unique decomposition of $\pi$ as a product of disjoint cycles. In practice, zero factors are omitted, and every factor of the form $\ell^1$ is written $\ell$. Thus, for example, $z_{(123)(456)(7)}=3^21$. Then, the cycle structure of a paratopism $\Theta=(\pi;(f_1,f_2,f_3))\in S_3\wr S_n$ is the tuple $z_\Theta:=(z_{\pi};(z_{f_1},z_{f_2},z_{f_3}))$. Mendis and Wanless \cite{Mendis2017} determined the cycle structures of all autoparatopisms of types I and II of Latin squares of order up to $17$. Those of all autotopisms of Latin squares of the same orders were already known \cite{Falcon2012, Stones2012}.

By Lemma \ref{lemma_uc}, the set of cycle structures of paratopisms in $\mathfrak{P}(n)$ that are conjugate to at least one paratopism in $\mathrm{Apar}(L)/\sim$ depends only on the main class of $L$. For each cycle structure $z$, we have calculated all feasible sizes of every critical set of $L$ based on a paratopism in $\mathrm{Apar}(L)/\sim$ that is conjugate to one in $\mathfrak{P}(n)$ with cycle structure $z$. These sizes were obtained in \cite{Falcon2021} for all autotopisms of Latin squares of order $n\leq 5$. We have confirmed all these values except for the largest size of a critical set based on any autotopism of $L_{5.2}$ with the cycle structure $(1^3; (31^2,31^2,31^2))$. This value is five instead of six, which was indicated in Example 32 in \cite{Falcon2021}. 

\section{A new secret sharing scheme}\label{sec:secret}

In this section, critical sets based on paratopisms are implemented to design a new secret sharing scheme customized for $\ell$ participants that is not so dependent on the existence of absent but indispensable participants. As a first approach, the dealer considers a pair $(L,\Theta)\in\mathcal{L}(n)\times \mathrm{Apar}(L)$ such that there is some critical set $P\in\mathrm{CS}_\Theta(L)$ of size $\ell$. Similarly to the scheme proposed in \cite{Cooper1994}, the dealer distributes to each participant a distinct entry of $P$ so that all of them must collaborate to recover the Latin square $L$, which is the key.

\begin{example}\label{ejemplo:nuevoEsquemaComparticion} For $\ell=3$, the dealer can consider the autoparatopism 
$\Theta= (\mathrm{Id}_3;((2354),(1243),(1243)))$ of the following $\Theta$-colored Latin square.

\[\L\equiv\begin{array}{|c|c|c|c|c|}\hline
        \cellcolor{red!100}{\color{white} 1} & \cellcolor{red!100}{\color{white} 2}  & \cellcolor{red!100}{\color{white} 3} &
        \cellcolor{red!100}{\color{white} 4} &
        \cellcolor{pink!100}{\color{white} 5}
        \\ \hline
        \cellcolor{blue!100}{\color{white} 2} &
        \cellcolor{violet!100}{\color{white} 3} & \cellcolor{magenta!100}{\color{white} 4}  & \cellcolor{orange!100}{\color{white} 5} &
        \cellcolor{teal!100}{\color{white} 1}
        \\ \hline
        \cellcolor{magenta!100}{\color{white} 3} &
        \cellcolor{blue!100}{\color{white} 4} &
        \cellcolor{orange!100}{\color{white} 5} & \cellcolor{violet!100}{\color{white} 1}  & \cellcolor{teal!100}{\color{white} 2}
        \\ \hline
        \cellcolor{violet!100}{\color{white} 4} &
        \cellcolor{orange!100}{\color{white} 5} &
        \cellcolor{blue!100}{\color{white} 1} &
        \cellcolor{magenta!100}{\color{white} 2} & \cellcolor{teal!100}{\color{white} 3}
        \\ \hline
        \cellcolor{orange!100}{\color{white} 5} &
        \cellcolor{magenta!100}{\color{white} 1} &
        \cellcolor{violet!100}{\color{white} 2} &
        \cellcolor{blue!100}{\color{white} 3} &
        \cellcolor{teal!100}{\color{white} 4}
        \\ \hline
        \end{array}\]
A $\Theta$-critical set of $L$ of size three is formed by the entries $\lbrace(2,2,3),(2,4,5),$ $(2,5,1)\rbrace$.\hfill $\lhd$
\end{example}

As a second approach, and again similarly to the scheme proposed in \cite{Cooper1994}, our protocol can be considered a multilevel scheme if the dealer considers more than one $\Theta$-critical set of the key under consideration, some of them even of distinct sizes. If their entries are distributed to the participants of this new scheme, those participants who have entries associated with longer orbits have more information about the key than those with entries in smaller orbits. In this regard, those orbits with a larger number of entries can be assigned to a higher level.

A third approach, which justifies the use of critical sets that are based on paratopisms, consists of distributing to the participants more than one entry of each $\Theta$-orbit. Those participants containing these entries play the same role in the scheme and hence, they become non-essential.

\begin{example} In a bank, there are four employees and one director. To open the safe, the requirement is that there must always be two employees and the director present. Therefore, they decide to divide the safe's key using a secret sharing scheme. The protocol determines that critical sets of three entries will be sought, one for the director and two for two of the employees, who will be the main participants of the scheme. The remaining employees will be considered substitutes for the participating employees, so redundant entries from the orbits used with the other employees will be assigned to them. As an illustrative example, we can consider the autoparatopism $\Theta=(1^3;((13)(45),(25)(34),(13)(45)))$ of the following $\Theta$-colored Latin square.
{\small \begin{center}
        $L\equiv{\footnotesize \begin{array}{|c|c|c|c|c|}\hline
        \cellcolor{red!100}{\color{white} 1} & 
        \cellcolor{purple!100}{\color{white} 2}  & 
        \cellcolor{orange!100}{\color{white} 3} &
        \cellcolor{green!100}{\color{black} 4} &
        \cellcolor{violet!100}{\color{white} 5}
        \\ \hline
        \cellcolor{magenta!100}{\color{white} 2} &
        \cellcolor{blue!100}{\color{white} 1} & 
        \cellcolor{brown!100}{\color{white} 4}  &
        \cellcolor{brown!100}{\color{white} 5} &
        \cellcolor{blue!100}{\color{white} 3}
        \\ \hline
        \cellcolor{red!100}{\color{white} 3} &
        \cellcolor{violet!100}{\color{white} 4} &
        \cellcolor{green!100}{\color{black} 5} & 
        \cellcolor{orange!100}{\color{white} 1}  & 
        \cellcolor{purple!100}{\color{white} 2}
        \\ \hline
        \cellcolor{pink!100}{\color{white} 4} &
        \cellcolor{teal!100}{\color{white} 5} &
        \cellcolor{yellow!100}{\color{white} 2} &
        \cellcolor{gray!100}{\color{white} 3} & 
        \cellcolor{olive!100}{\color{white} 1}
        \\ \hline
        \cellcolor{pink!100}{\color{white} 5} &
        \cellcolor{olive!100}{\color{white} 3} &
        \cellcolor{gray!100}{\color{white} 1} &
        \cellcolor{yellow!100}{\color{white} 2} &
        \cellcolor{teal!100}{\color{white} 4}
        \\ \hline
        \end{array}}$
    \end{center}}

A $\Theta$-critical set of $L$ of size three is formed by the entries $\lbrace(1,4,4),$ $(4,2,5),(4,4,3)\rbrace$. The entry $(1,4,4)$ is given to the director, while each of the other two entries is, respectively, given to one of the other two participants. In addition, the third and fourth employees, who are not main participants in the scheme, receive, respectively, the entries $(5,5,4)$ and $(5,3,1)$. The role of the director is indispensable because, if the entry $(1, 4, 4)$ is unknown, then the following alternative Latin square also has $\Theta$ as an autoparatopism.
{\small    \begin{center}
        ${\footnotesize \begin{array}{|c|c|c|c|c|}\hline
        \cellcolor{red!100}{\color{black} 4} & 
        \cellcolor{purple!100}{\color{black} 1}  & 
        \cellcolor{orange!100}{\color{black} 3} &
        \cellcolor{green!100}{\color{black} 2} &
        \cellcolor{violet!100}{\color{black} 5}
        \\ \hline
        \cellcolor{magenta!100}{\color{white} 2} &
        \cellcolor{blue!100}{\color{black} 3} & 
        \cellcolor{brown!100}{\color{black} 5}  &
        \cellcolor{brown!100}{\color{black} 4} &
        \cellcolor{blue!100}{\color{black} 1}
        \\ \hline
        \cellcolor{red!100}{\color{black} 5} &
        \cellcolor{violet!100}{\color{black} 4} &
        \cellcolor{green!100}{\color{black} 2} & 
        \cellcolor{orange!100}{\color{black} 1}  & 
        \cellcolor{purple!100}{\color{black} 3}
        \\ \hline
        \cellcolor{pink!100}{\color{black} 1} &
        \cellcolor{teal!100}{\color{white} 5} &
        \cellcolor{yellow!100}{\color{black} 4} &
        \cellcolor{gray!100}{\color{white} 3} & 
        \cellcolor{olive!100}{\color{black} 2}
        \\ \hline
        \cellcolor{pink!100}{\color{black} 3} &
        \cellcolor{olive!100}{\color{black} 2} &
        \cellcolor{gray!100}{\color{white} 1} &
        \cellcolor{yellow!100}{\color{black} 5} &
        \cellcolor{teal!100}{\color{white} 4}
        \\ \hline
        \end{array}}$
    \end{center}}

Furthermore, the first and second employees are not indispensable because they can, respectively, be replaced by the third and fourth employees. This is due to the fact that each substitute has an alternative entry of the same $\Theta$-orbit being considered. In this way, the director's entry is superior in rank because it constitutes a non-redundant, indispensable piece of information required to achieve unique completion.  \hfill $\lhd$
\end{example}

The security of these secret sharing schemes is primarily framed in terms of the brute-force search-space size. Note that the total number of Latin squares grows super-exponentially, with the known lower bound $|\mathcal{L}(n)|\geq \frac {\left(n!\right)^{2n}}{n^{n^2}}\geq n^{(n^2-o(n^2))}$. More precisely, the exact number is known for $n\leq 11$ \cite{Hulpke2011, Kolesova1990, McKay2007}. In particular, $|\mathcal{L}(10)|\approx 9.9824 \times 10^{36} \leq 2^{128}\leq 7.7696\times 10^{47}\approx|\mathcal{L}(11)|$. Since 128 bits of effective security is the current threshold to consider an absolutely secure system against traditional brute-force attacks, our schemes are adequate for $n\geq 11$. Our examples fall well short of this bound. They are not presented as real-world deployment examples, but only as a mere illustration of the described techniques. 

For an adversary who does not hold all the entries of a $\Theta$-critical set, the relevant hardness is completing a partial Latin square with the missing entries unknown. This problem remains NP-complete in the worst case \cite{Colbourn1984}, although specific instances may be solved considerably faster using standard heuristics. For any authorized coalition (holding all the entries of a critical set, and having $\Theta$ public as part of the protocol), recovery does not require brute-force search over all Latin squares of order $n$. Since $\Theta$ is known, the coalition immediately fills every entry in $\mathrm{Orb}_\Theta(e)$ for each known entry $e$.

\section{Conclusion and further work}

In this paper, we have introduced the problem of determining critical sets of Latin squares containing a given non-trivial paratopism in their group of autoparatopisms. We have proved that the feasible sizes of these critical sets depend only on the conjugacy class of the paratopism and the main class of the Latin square under consideration. A pair of algorithms has been described to compute these sizes for Latin squares. As an illustrative example, we have implemented them to determine the sizes of critical sets based on autoparatopisms of Latin squares of order up to six. 

We have also described a new secret sharing scheme arising from critical sets based on an autoparatopism of a given Latin square, which is the key. This cryptographic protocol reduces the complexity when searching for critical sets for a given Latin square by incorporating an autoparatopism that determines the orbits of the entries within that set. It reduces the number of assumptions that need to be made for the remaining empty entries. It also allows for the use of multilevel secret sharing schemes without the need to employ more than one critical set. This is because the use of orbits of the elements related to the mentioned set allows us to assign entries related to longer orbits (more information) to individuals with higher rank, or alternatively, reduce the number of entries that are distributed to assign redundant entries to the substitutes.

\vspace{0.25cm}

\section*{Acknowledgements}

This paper has partially been supported by the Research and Innovation Project PPIT-FEDER 2023  {\em ``Modeling small-world, scale-free networks from combinatorial designs based on quasigroup digraphs''},  co-financed by the EU - Ministry of Finance and Public Administration - European Funds - Andalusian Regional Government - Ministry of University, Research and Innovation.

\newpage

\begin{appendices}

\section{Bounds for critical sets}\label{secA}

This appendix collects our computational results of implementing Algorithm \ref{alg_2} in {\sc GAP} concerning the smallest and largest sizes of critical sets based on autoparatopisms of Latin squares of order $n\in\{3,4,5,6\}$. 

\begin{table}[ht]
    \caption{Smallest and largest sizes of critical sets based on isotopisms.}
    \label{Table_scs_lcs}
    \footnotesize
    \renewcommand{\tabcolsep}{3pt}
    \centering
{\tiny   \begin{tabular}{lllccr||lllccr}\toprule
    $L$ & $z$ & class & $\mathrm{scs}$ & $\mathrm{lcs}$ & time $(s)$ &
    $L$ & $z$ & class & $\mathrm{scs}$ & $\mathrm{lcs}$ & time $(s)$\\ \midrule

    $L_3$ & $(1^3,1^3,1^3)$ & $c_{3.1}$ & 2 & 3 & 0.06 &
    $L_{6.4}$ & $(1^6,3^2,3^2)$ & $c_{6.4.3}$ & 8 & 8 & 0.11\\
         & $(1^3,3,3)$ & $c_{3.2}$ & 2 & 2 & 0.00 &
         & $(1^6,6,6)$ & $c_{6.4.4}$ & 5 & 5 & 0.00\\
         & $(21,21,21)$ & $c_{3.3}$ & 1 & 2 & 0.00 &
         & $(2^21^2,2^21^2,2^21^2)$ & $c_{6.4.5}$ & 7 & 9 & 114.09\\
         & $(3,3,3)$ & $c_{3.4}$ & 1 & 1 & 0.00 &
         & $(2^21^2,2^3,2^3)$ & $c_{6.4.6}$ & 6 & 8 & 34.78\\
    $L_{4.1}$ & $(1^4,1^4,1^4)$ & $c_{4.1.1}$ & 5 & 7 & 1.72 &
         & $(3^2,3^2,3^2)$ & $c_{6.4.7}$ & 4 & 5 & 0.42\\
         & $(1^4,2^2,2^2)$ & $c_{4.1.2}$ & 4 & 4 & 0.02 &
         & $(3^2,6,2^3)$ & $c_{6.4.8}$ & 4 & 4 & 0.00\\
         & $(21^2,21^2,21^2)$ & $c_{4.1.3}$ & 4 & 4 & 0.03 &
         & $(3^2,6,6)$ & $c_{6.4.9}$ & 3 & 3 & 0.02\\
         & $(21^2,4,4)$ & $c_{4.1.4}$ & 2 & 2 & 0.00 &
    $L_{6.5}$ & $(1^6,1^6,1^6)$ & $c_{6.5.1}$ & 10 & 17 & $>$36,000.00\\
         & $(2^2,2^2,2^2)$ & $c_{4.1.5}$ & 3 & 3 & 0.00 &
         & $(1^6,2^3,2^3)$ & $c_{6.5.2}$ & 9 & 10 & 22.70\\
         & $(31,31,31)$ & $c_{4.1.6}$ & 2 & 2 & 0.00 &
         & $(2^21^2,2^21^2,2^21^2)$ & $c_{6.5.3}$ & 7 & 9 & 120.77\\
    $L_{4.2}$ & $(1^4,1^4,1^4)$ & $c_{4.2.1}$ & 4 & 6 & 1.37 &
         & $(2^21^2,2^3,2^3)$ & $c_{6.5.4}$ & 6 & 8 & 32.84\\
         & $(1^4,2^2,2^2)$ & $c_{4.2.2}$ & 4 & 4 & 0.01 &
         & $(31^3,3^2,3^2)$ & $c_{6.5.5}$ & 4 & 6 & 0.53\\
         & $(1^4,4,4)$ & $c_{4.2.3}$ & 3 & 3 & 0.00 &
         & $(31^3,6,6)$ & $c_{6.5.6}$ & 3 & 3 & 0.01\\
         & $(21^2,21^2,21^2)$ & $c_{4.2.4}$ & 4 & 4 & 0.03 &
    $L_{6.6}$ & $(1^6,1^6,1^6)$ & $c_{6.6.1}$ & 11 & 17 & $>$36,000.00\\
         & $(21^2,2^2,2^2)$ & $c_{4.2.5}$ & 3 & 3 & 0.01 &
         & $(21^4,2^3,2^3)$ & $c_{6.6.2}$ & 7 & 9 & 36.36\\
         & $(2^2,4,4)$ & $c_{4.2.6}$ & 2 & 2 & 0.00 &
         & $(2^21^2,2^21^2,2^21^2)$ & $c_{6.6.3}$ & 7 & 9 & 117.95\\
    $L_{5.1}$ & $(1^5,1^5,1^5)$ & $c_{5.1.1}$ & 6 & 10 & 3,483.17 &
         & $(31^3,3^2,3^2)$ & $c_{6.6.4}$ & 4 & 6 & 0.53\\
         & $(5,5,1^5)$ & $c_{5.1.2}$ & 4 & 4 & 0.00 &
         & $(321,6,6)$ & $c_{6.6.5}$ & 2 & 3 & 0.01\\
         & $(2^21,2^21,2^21)$ & $c_{5.1.3}$ & 3 & 5 & 0.36 &
         & $(41^2,41^2,41^2)$ & $c_{6.6.6}$ & 5 & 6 & 0.45\\
         & $(41,41,41)$ & $c_{5.1.4}$ & 3 & 3 & 0.00 &
         & $(51,51,51)$ & $c_{6.6.7}$ & 3 & 4 & 0.05\\
         & $(5,5,5)$ & $c_{5.1.5}$ & 2 & 2 & 0.00 &
    $L_{6.7}$ & $(1^6,1^6,1^6)$ & $c_{6.7.1}$ & 11 & 17 & $>$36,000.00\\
    $L_{5.2}$ & $(1^5,1^5,1^5)$ & $c_{5.2.1}$ & 7 & 11 & 4,069.39 &
         & $(1^6,2^3,2^3)$ & $c_{6.7.2}$ & 9 & 10 & 9.45\\
         & $(2^21,2^21,2^21)$ & $c_{5.2.2}$ & 3 & 5 & 0.45 &
         & $(2^21^2,2^21^2,2^21^2)$ & $c_{6.7.3}$ & 7 & 9 & 113.36\\
         & $(31^2,31^2,31^2)$ & $c_{5.2.3}$ & 5 & 5\footnotemark[1] & 0.09 &
         & $(2^3,2^3,2^21^2)$ & $c_{6.7.4}$ & 6 & 8 & 33.03\\
    $L_{6.1}$ & $(1^6,1^6,1^6)$ & $c_{6.1.1}$ & 11 & 17 & $>$36,000.00 &
         & $(3^2,3^2,3^2)$ & $c_{6.7.5}$ & 4 & 5 & 0.45\\
         & $(1^6,2^3,2^3)$ & $c_{6.1.2}$ & 9 & 10 & 36.87 &
         & $(3^2,6,6)$ & $c_{6.7.6}$ & 3 & 3 & 0.00\\
         & $(2^21^2,2^21^2,2^21^2)$ & $c_{6.1.3}$ & 7 & 10 & 104.11 &
    $L_{6.8}$ & $(1^6,1^6,1^6)$ & $c_{6.8.1}$ & 10 & 17 & $>$36,000.00\\
         & $(2^21^2,2^3,2^3)$ & $c_{6.1.4}$ & 7 & 8 & 27.64 &
         & $(21^4,2^3,2^3)$ & $c_{6.8.2}$ & 7 & 10 & 31.92\\
    $L_{6.2}$ & $(1^6,1^6,1^6)$ & $c_{6.2.1}$ & 12 & 18 & $>$36,000.00 &
         & $(2^21^2,2^21^2,2^21^2)$ & $c_{6.8.3}$ & 7 & 9 & 119.67\\
         & $(1^6,2^3,2^3)$ & $c_{6.2.2}$ & 9 & 10 & 57.44 &
         & $(41^2,41^2,41^2)$ & $c_{6.8.4}$ & 5 & 6 & 0.48\\
         & $(1^6,3^2,3^2)$ & $c_{6.2.3}$ & 8 & 8 & 0.23 &
    $L_{6.9}$ & $(1^6,1^6,1^6)$ & $c_{6.9.1}$ & 10 & 17 & $>$36,000.00\\
         & $(2^21^2,2^21^2,2^21^2)$ & $c_{6.2.4}$ & 7 & 10 & 133.55 &
         & $(2^21^2,2^3,2^3)$ & $c_{6.9.2}$ & 6 & 8 & 33.30\\
         & $(2^3,2^3,2^21^2)$ & $c_{6.2.5}$ & 7 & 9 & 64.95 &
    $L_{6.10}$ & $(1^6,1^6,1^6)$ & $c_{6.10.1}$ & 10 & 17 & $>$36,000.00\\
         & $(2^3,3^2,6)$ & $c_{6.2.6}$ & 4 & 4 & 0.02 &
         & $(21^4,2^3,2^3)$ & $c_{6.10.2}$ & 7 & 9 & 32.16\\
         & $(31^3,31^3,31^3)$ & $c_{6.2.7}$ & 8 & 9 & 23.59 &
         & $(2^21^2,2^21^2,2^21^2)$ & $c_{6.10.3}$ & 7 & 9 & 122.47\\
         & $(3^2,3^2,31^3)$ & $c_{6.2.8}$ & 6 & 7 & 0.22 &
    $L_{6.11}$ & $(1^6,1^6,1^6)$ & $c_{6.11.1}$ & 10 & 17 & $>$36,000.00\\
         & $(31^3,6,6)$ & $c_{6.2.9}$ & 3 & 3 & 0.03 &
         & $(2^21^2,2^21^2,2^21^2)$ & $c_{6.11.2}$ & 7 & 9 & 123.72\\
         & $(6,6,2^21^2)$ & $c_{6.2.10}$ & 4 & 4 & 0.00 &
    $L_{6.12}$ & $(1^6,1^6,1^6)$ & $c_{6.12.1}$ & 11 & 17 & $>$36,000.00\\
    $L_{6.3}$ & $(1^6,1^6,1^6)$ & $c_{6.3.1}$ & 11 & 16 & $>$36,000.00 &
         & $(1^6,3^2,3^2)$ & $c_{6.12.2}$ & 8 & 8 & 0.12\\
         & $(1^6,3^2,3^2)$ & $c_{6.3.2}$ & 8 & 8 & 0.22 &
         & $(21^4,2^3,2^3)$ & $c_{6.12.3}$ & 7 & 9 & 32.52\\
         & $(2^21^2,2^3,2^3)$ & $c_{6.3.3}$ & 6 & 8 & 43.19 &
         & $(21^4,6,6)$ & $c_{6.12.4}$ & 4 & 5 & 0.00\\
         & $(2^21^2,6,6)$ & $c_{6.3.4}$ & 4 & 4 & 0.01 &
         & $(2^21^2,2^21^2,2^21^2)$ & $c_{6.12.5}$ & 7 & 9 & 129.17\\
         & $(31^3,31^3,31^3)$ & $c_{6.3.5}$ & 8 & 9 & 20.41 &
         & $(31^3,31^3,31^3)$ & $c_{6.12.6}$ & 8 & 9 & 31.50\\
         & $(3^2,3^2,31^3)$ & $c_{6.3.6}$ & 6 & 6 & 0.20 &
         & $(31^3,3^2,3^2)$ & $c_{6.12.7}$ & 6 & 7 & 0.67\\
    $L_{6.4}$ & $(1^6,1^6,1^6)$ & $c_{6.4.1}$ & 9 & 17 & $>$36,000.00 &
         & $(321,6,6)$ & $c_{6.12.8}$ & 3 & 3 & 0.00\\
         & $(1^6,2^3,2^3)$ & $c_{6.4.2}$ & 9 & 10 & 9.14 &
         & & & & & \\ \hline
    \end{tabular}}
    \footnotetext[1]{This corrects a wrong value given in \cite{Falcon2021}.}
\end{table}

\begin{table}[ht]
\caption{Smallest and largest sizes of critical sets based on paratopisms of Type I.}\label{table1}%
    \footnotesize
    \renewcommand{\tabcolsep}{3pt}
    \centering
\begin{tabular}{lllccr||lllccr}
\toprule
$L$ & $z$  & class & $\mathrm{scs}$ & $\mathrm{lcs}$ & time $(s)$ & $L$ & $z$ & class  & $\mathrm{scs}$ & $\mathrm{lcs}$ & time $(s)$\\
\midrule
$L_3$ & $(1^3,1^3,1^3)$  & $c_{3.5}$ & $2$ & $2$ & 0.00 & $L_{6.4}$ &  $(1^6,1^6,2^21^2)$ & $c_{6.4.11}$& 6 & 9 & 255.98\\
 &  $(1^3,1^3,21)$  & $c_{3.6}$ & $1$ & $1$&  0.00 & &  $(1^6,2^3,2^3)$ & $c_{6.4.12}$& 6 & 6 & 0.05\\
 & $(1^3,3,21)$  & $c_{3.7}$ & $1$ & $1$& 0.00 & &  $(1^6,3^2,2^21^2)$& $c_{6.4.13}$ & 4 & 5 & 0.00\\
 & $(1^3,3,3)$  & $c_{3.8}$ & $1$ & $1$& 0.00&  &  $(1^6,3^2,3^2)$& $c_{6.4.14}$ & 3 & 3 &  0.02 \\
$L_{4.1}$ & $(1^4,1^4,1^4)$  & $c_{4.1.7}$ & $4$ & $6$& 0.03 &  &  $(1^6,6,2^3)$& $c_{6.4.15}$& 3 & 3 &0.00 \\
 & $(1^4,1^4,21^2)$  & $c_{4.1.8}$& $3$ & $4$& 0.03  & &  $(1^6,6,6)$& $c_{6.4.16}$ & 2 & 2 & 0.00 \\
 & $(1^4,2^2,21^2)$  & $c_{4.1.9}$& $3$ & $3$& 0.00 & $L_{6.5}$  &  $(1^6,1^6,21^4)$& $c_{6.5.7}$ & 7 & 10  & 242.23\\
 & $(2^2,2^2,1^4)$  & $c_{4.1.10}$& $3$ & $3$& 0.00 &
 &  & $c_{6.5.8}$ & 6 & 10  & 263.39\\
 & $(1^4,2^2,4)$  & $c_{4.1.11}$& $2$ & $2$& 0.00 & &  $(1^6,3^2,321)$& $c_{6.5.9}$& 4 & 4 & 0.02\\
 & $(1^4,31,31)$  & $c_{4.1.12}$& $1$ & $1$& 0.00 & & & $c_{6.5.10}$& 2 & 4\footnotemark[1] & 0.02\\
$L_{4.2}$ & $(1^4,1^4,1^4)$  & $c_{4.2.7}$ & $3$ & $4$& 0.03 & $L_{6.6}$  &  $(1^6,1^6,1^6)$& $c_{6.6.8}$& 8 & 15\footnotemark[1] & 235.75\\
 & $(1^4,1^4,21^2)$  & $c_{4.2.8}$& $3$ & $4$&0.02 & &  $(1^6,1^6,21^4)$ & $c_{6.6.9}$ & 7  & 11 & 258.42\\
 & $(1^4,2^2,21^2)$ & $c_{4.2.9}$ & $3$ & $3$&0.00 & &  $(1^6,1^6,2^21^2)$ & $c_{6.6.10}$ & 6 & 9 & 265.03\\
 & $(1^4,2^2,2^2)$& $c_{4.2.10}$& $3$ & $3$&0.00  & &  $(1^6,2^21^2,41^2)$ & $c_{6.6.11}$ & 5 & 6 & 0.20\\
 & $(1^4,4,2^2)$ & $c_{4.2.11}$& $2$ & $2$&0.00 & &  $(1^6,3^2,31^3)$ & $c_{6.6.12}$ & 3  & 4 & 0.01\\
 & $(1^4,4,4)$ & $c_{4.2.12}$& $1$ & $1$&0.00 & &  $(1^6,3^2,321)$ & $c_{6.6.13}$ & 4  & 4 & 0.02\\
$L_{5.1}$ & $(1^5,1^5,1^5)$& $c_{5.1.6}$ & 4 & 6 & 1.86 & &  $(1^6,51,51)$ & $c_{6.6.14}$ & 1 & 1 & 0.00\\
 & $(1^5,1^5,2^21)$& $c_{5.1.7}$ & 3 & 4 & 0.59 & $L_{6.7}$ & $(1^6,1^6,1^6)$& $c_{6.7.7}$ & 7 & 10 & 167.67 \\
 & $(1^5,2^21,41)$& $c_{5.1.8}$ & 3 & 3 & 0.02 & & $(1^6,1^6,2^21^2)$& $c_{6.7.8}$  & 6 & 9 & 215.67\\
 & $(1^5,5,2^21)$& $c_{5.1.9}$ & 2 & 2 & 0.00 & & $(1^6,2^3,2^3)$& $c_{6.7.9}$& 6 & 6 & 0.05\\
 & $(1^5,5,5)$& $c_{5.1.10}$& 1 & 1 & 0.00 & & $(1^6,3^2,3^2)$& $c_{6.7.10}$& 3 & 3 & 0.02\\
$L_{5.2}$ & $(1^5,1^5,21^3)$& $c_{5.2.4}$ & 4 & 6 & 1.59 & & $(1^6,6,6)$& $c_{6.7.11}$ & 2 & 2 & 0.00\\
 & $(1^5,2^21,41)$& $c_{5.2.5}$ & 3 & 3 & 0.00 & $L_{6.8}$  &  $(1^6,1^6,1^6)$& $c_{6.8.5}$ & 7 & 11 & 279.34\\
$L_{6.1}$ & $(1^6,1^6,21^4)$& $c_{6.1.5}$ & 6 & 9 & 446.11 & & $(1^6,1^6,21^4)$& $c_{6.8.6}$ & 7 & 9 & 252.94\\
 & & $c_{6.1.6}$  & 7 & 10 & 146.87 & & $(1^6,1^6,2^21^2)$& $c_{6.8.7}$ & 6 & 9 & 214.11\\
 & $(1^6,2^3,42)$& $c_{6.1.7}$ & 4 & 4 & 0.03 & & $(1^6,6,51)$& $c_{6.8.8}$ & 5 & 6 & 0.22\\
 & & $c_{6.1.8}$  & 5 & 5 & 0.05 & $L_{6.9}$ & $(1^6,1^6,21^4)$& $c_{6.9.3}$  & 6  & 10   & 241.41 \\
$L_{6.2}$ & $(1^6,1^6,21^4)$& $c_{6.2.11}$  & 7 & 10 & 247.89 & & $(1^6,2^3,42)$ & $c_{6.9.4}$  & 4  & 5   &  0.06\\
 & & $c_{6.2.12}$  & 8 & 11 & 175.30 & $L_{6.10}$ & $(1^6,1^6,1^6)$ & $c_{6.10.4}$ & 6  & 9   & 252.11\\
 & $(1^6,2^3,42)$& $c_{6.2.13}$ & 4 & 4 & 0.05 & & $(1^6,1^6,21^4)$& $c_{6.10.5}$ & 6  & 10   & 281.40\\
 & & $c_{6.2.14}$  & 5 & 5 & 0.03 & & & $c_{6.10.6}$ & 7  & 10   & 289.14\\
 & $(1^6,31^3,321)$& $c_{6.2.15}$ & 4 & 4 & 0.17 & & $(1^6,1^6,2^21^2)$& $c_{6.10.7}$ & 7  & 10   & 251.95\\
 & $(1^6,3^2,321)$ & $c_{6.2.16}$ & 4 & 4 & 0.00 & $L_{6.11}$ & $(1^6,1^6,21^4)$& $c_{6.11.3}$ & 6 & 10 & 182.61\\
 & & $c_{6.2.17}$ & 4 & 5 & 0.01 & & $(1^6,2^21^2,41^2)$& $c_{6.11.4}$ & 5 & 6 & 0.19\\
 & $(1^6,3^2,21^4)$& $c_{6.2.18}$ & 5 & 5 & 0.02 & $L_{6.12}$ & $(1^6,1^6,1^6)$ & $c_{6.12.9}$ & 7 & 10 & 243.44\\
 & $(1^6,6,42)$& $c_{6.2.19}$ & 2 & 2 & 0.00 & & $(1^6,1^6,21^4)$& $c_{6.12.10}$& 6 & 10 & 314.14\\
$L_{6.3}$ & $(1^6,1^6,21^4)$& $c_{6.3.7}$& 6 & 9 & 275.06 & & & $c_{6.12.11}$ & 7 & 11\footnotemark[1] & 261.97\\
 & $(1^6,2^3,42)$ & $c_{6.3.8}$ & 4 & 5 & 0.00 & & $(1^6,1^6,2^21^2)$ & $c_{6.12.12}$ & 7 & 10 & 276.50\\
 & $(1^6,3^2,321)$ & $c_{6.3.9}$& 4 & 4 & 0.00 & & $(1^6,31^3,321)$&$c_{6.12.13}$& 4 & 4 & 0.20\\
 & $(1^6,6,42)$ & $c_{6.3.10}$ & 2 & 2 & 0.00 & & $(1^6,3^2,21^4)$&$c_{6.12.14}$& 5 & 5 & 0.02\\
$L_{6.4}$ & $(1^6,1^6,1^6)$ & $c_{6.4.10}$& 6 & 10 & 184.14 & & $(1^6,3^2,31^3)$ & $c_{6.12.15}$& 4 & 4 & 0.01\\
 &  & &  &  &  & & $(1^6,3^2,2^21^2)$ & $c_{6.12.16}$& 4 & 5 & 0.02\\
 &  & &  &  &  & & $(1^6,3^2,321)$ & $c_{6.12.17}$& 4 & 5 & 0.02\\
\hline
\end{tabular}
\footnotetext[1]{There is a gap in the sizes.}
\end{table}

\clearpage

\begin{table}[ht]
\caption{Smallest and largest sizes of critical sets based on paratopisms of Type II.}\label{table2}%
    \footnotesize
    \renewcommand{\tabcolsep}{3pt}
    \centering
\begin{tabular}{lllccr||lllccr}
\toprule
$L$ & $z$ & class  & $\mathrm{scs}$ & $\mathrm{lcs}$ &time $(s)$ & $L$ & $z$ & class  & $\mathrm{scs}$ & $\mathrm{lcs}$ &time $(s)$ \\
\midrule
$L_3$ & $(1^3,1^3,1^3)$ & $c_{3.9}$ & $1$ & $1$ & 0.00  &  $L_{6.1}$  & $(1^6,1^6,1^6)$ & $c_{6.1.9}$ & 4 & 7 & 1.81\\
& $(1^3,1^3,21)$ & $c_{3.10}$ & $0$ & $0$ & 0.00 & & $(1^6,1^6,2^21^2)$ & $c_{6.1.10}$ & 3 & 4 & 0.03\\
$L_{4.1}$ & $(1^4,1^4,1^4)$ & $c_{4.1.13}$ & $3$ & $3$ &0.00 & $L_{6.2}$   & $(1^6,1^6,1^6)$ & $c_{6.2.20}$ & 5 & 7 & 1.80\\
 &  & $c_{4.1.14}$ & $2$ & $3$& 0.02   & & $(1^6,1^6,2^21^2)$ & $c_{6.2.21}$ & 3 & 4 & 0.03\\
 & $(1^4,1^4,21^2)$ & $c_{4.1.15}$ & $2$ & $2$&0.00 &  & $(1^6,1^6,31^3)$ & $c_{6.2.22}$ & 3 & 3 & 0.00\\ 
 & $(1^4,1^4,2^2)$ & $c_{4.1.16}$ & $1$ & $1$& 0.00 & $L_{6.3}$ & $(1^6,1^6,1^6)$ & $c_{6.3.11}$ & 4 & 7 & 0.80\\
$L_{4.2}$ & $(1^4,1^4,1^4)$ & $c_{4.2.13}$ & 2 & 3  & 0.02& & $(1^6,1^6,31^3)$ & $c_{6.3.12}$ & 3 & 3 & 0.02\\
  & $(1^4,1^4,21^2)$ & $c_{4.2.14}$ & 2 & 2& 0.00& $L_{6.4}$  & $(1^6,1^6,1^6)$ &  $c_{6.4.17}$ & 3 & 5\footnotemark[1] & 0.52\\
$L_{5.1}$ & $(1^5,1^5,1^5)$ & $c_{5.1.11}$ & 2 & 4 & 0.01&  && $c_{6.4.18}$ & 5 & 6 & 1.55\\
 & $(1^5,1^5,2^21)$ & $c_{5.1.12}$ & 1 & 1 & 0.00& &$(1^6,1^6,2^21^2)$ & $c_{6.4.19}$ & 3  & 4  & 0.03 \\
  & $(1^5,1^5,41)$ & $c_{5.1.13}$ & 1 & 1 & 0.00 & $L_{6.7}$ & $(1^6,1^6,1^6)$ & $c_{6.7.12}$ & 5 & 7 & 1.62\\
$L_{5.2}$ & $(1^5,1^5,1^5)$ & $c_{5.2.6}$ & 3 & 5 & 0.03 & &$(1^6,1^6,2^21^2)$ & $c_{6.7.13}$ & 3 & 4 & 0.03\\
     &  & $c_{5.2.7}$ & 3 & 6 & 0.09&$L_{6.9}$ & $(1^6,1^6,1^6)$ & $c_{6.9.5}$ & 4 & 8 & 1.94\\
 & & $c_{5.2.8}$  & 4 & 5& 0.02&$L_{6.11}$ & $(1^6,1^6,1^6)$ & $c_{6.11.5}$ & 4 & 7 & 1.87\\
 & $(1^5,1^5,2^21)$ & $c_{5.2.9}$ & 1 & 1 & 0.00\\
\hline
\end{tabular}
\footnotetext[1]{There is a gap in the sizes.}
\end{table}

\section{Examples of critical sets based on non-trivial autoparatopisms}\label{secA1}

This appendix provides examples of critical sets associated with non-trivial autoparatopisms of the Latin squares considered in this paper. For each conjugacy class listed in Tables~\ref{Table_scs_lcs}--\ref{table2}, we provide a representative non-trivial autoparatopism together with one critical set for each cardinality occurring in that class. The labels coincide with those used in the main tables. Furthermore, we represent each entry $(r,c,s)$ in a critical set as $rcs$.

\renewcommand{\tabcolsep}{6pt}
\begin{table}[htbp]
\caption{Critical sets based on non-trivial isotopisms (I).}
    \label{Table_Atop2_app}
\tiny
\centering
\begin{tabularx}{\textwidth}{l l X c X}
\toprule
$L$ & Class & Autotopism & Size & Critical set \\
\midrule
$L_3$ & $c_{3.2}$ & $(\mathrm{Id}_3;(\mathrm{Id}_3,(123),(123)))$ & 2 & $\{111,212\}$ \\
      & $c_{3.3}$ & $(\mathrm{Id}_3;((23),(23),(23)))$ & 1 & $\{223\}$ \\
      &           &                                    & 2 & $\{122,212\}$ \\
      & $c_{3.4}$ & $(\mathrm{Id}_3;((123),(123),(132)))$ & 1 & $\{111\}$ \\
\midrule
$L_{4.1}$ & $c_{4.1.2}$ & $(\mathrm{Id}_3;(\mathrm{Id}_4,(12)(34),(12)(34)))$ & 4 & $\{111,133,313,432\}$ \\
          & $c_{4.1.3}$ & $(\mathrm{Id}_3;((34),(34),(34)))$ & 4 & $\{111,133,313,331\}$ \\
          & $c_{4.1.4}$ & $(\mathrm{Id}_3;((34),(1324),(1324)))$ & 2 & $\{111,313\}$ \\
          & $c_{4.1.5}$ & $(\mathrm{Id}_3;((12)(34),(13)(24),(14)(23)))$ & 3 & $\{111,122,313\}$ \\
          & $c_{4.1.6}$ & $(\mathrm{Id}_3;((243),(243),(243)))$ & 2 & $\{122,212\}$ \\
\midrule
$L_{4.2}$ & $c_{4.2.2}$ & $(\mathrm{Id}_3;(\mathrm{Id}_4,(12)(34),(12)(34)))$ & 4 & $\{111,133,313,431\}$ \\
          & $c_{4.2.3}$ & $(\mathrm{Id}_3;(\mathrm{Id}_4,(1324),(1324)))$ & 3 & $\{111,212,313\}$ \\
          & $c_{4.2.4}$ & $(\mathrm{Id}_3;((34),(34),(34)))$ & 4 & $\{111,133,313,332\}$ \\
          & $c_{4.2.5}$ & $(\mathrm{Id}_3;((34),(13)(24),(13)(24)))$ & 3 & $\{111,122,332\}$ \\
          & $c_{4.2.6}$ & $(\mathrm{Id}_3;((12)(34),(1324),(1423)))$ & 2 & $\{111,313\}$ \\
\midrule
$L_{5.1}$ & $c_{5.1.2}$ & $(\mathrm{Id}_3;(\mathrm{Id}_5,(12345),(12345)))$ & 4 & $\{111,212,313,414\}$ \\
          & $c_{5.1.3}$ & $(\mathrm{Id}_3;((25)(34),(25)(34),(25)(34)))$ & 3 & $\{223,245,352\}$ \\
          &             &                                                 & 4 & $\{122,212,223,335\}$ \\
          &             &                                                 & 5 & $\{122,133,212,223,313\}$ \\
          & $c_{5.1.4}$ & $(\mathrm{Id}_3;((2354),(2354),(2354)))$ & 3 & $\{122,212,223\}$ \\
          & $c_{5.1.5}$ & $(\mathrm{Id}_3;((12345),(12345),(13524)))$ & 2 & $\{111,122\}$ \\
\midrule
$L_{5.2}$ & $c_{5.2.2}$ & $(\mathrm{Id}_3;((13)(45),(25)(34),(13)(45)))$ & 3 & $\{133,155,451\}$ \\
&  &  & 4 & $\{111,122,144,451\}$ \\
&  &  & 5 & $\{111,122,133,425,432\}$ \\
 & $c_{5.2.3}$ & $(\mathrm{Id}_3;((345),(345),(345)))$ & 5 & $\{111,133,234,313,352\}$ \\ \midrule$L_{6.1}$ & $c_{6.1.2}$ & $(\mathrm{Id}_3;(\mathrm{Id}_6,(12)(35)(46),(12)(35)(46)))$ & 9 & $\{111,133,144,234,313,341,414,531,645\}$\\
&&&10& $\{111,133,144,234,243,313,336,442,616,632\}$\\
 & $c_{6.1.3}$ & $(\mathrm{Id}_3;((36)(45),(36)(45),(36)(45)))$ & 7 & $\{111,133,313,341,354,435,453\}$\\
&&&8& $\{111,133,144,234,313,341,435,453\}$\\
&&&9& $\{111,133,144,234,313,325,414,442,453\}$\\
&&&10& $\{111,133,144,234,325,362,414,426,442,461\}$\\
 & $c_{6.1.4}$ & $(\mathrm{Id}_3;((36)(45),(12)(34)(56),(12)(34)(56)))$ & 7 & $\{111,133,155,313,336,341,453\}$\\
&&&8& $\{111,133,155,313,325,336,453,461\}$\\ \midrule
$L_{6.2}$ & $c_{6.2.2}$ & $(\mathrm{Id}_3;(\mathrm{Id}_6,(14)(25)(36),(14)(25)(36)))$ & 9 & $\{111,122,133,212,223,313,435,524,634\}$\\
&&&10& $\{111,122,133,212,223,231,313,321,435,524\}$\\
 & $c_{6.2.3}$ & $(\mathrm{Id}_3;(\mathrm{Id}_6,(123)(465),(123)(465)))$ & 8 & $\{111,144,212,245,414,441,515,643\}$\\
 & $c_{6.2.4}$ & $(\mathrm{Id}_3;((23)(56),(23)(56),(23)(56)))$ & 7 & $\{111,122,212,231,256,536,563\}$\\
 &&&8& $\{111,122,155,212,264,426,536,563\}$\\
&&&9& $\{111,122,155,212,245,426,515,551,563\}$\\
&&&10& $\{111,122,155,212,264,453,515,524,542,551\}$\\
 & $c_{6.2.5}$ & $(\mathrm{Id}_3;((14)(26)(35),(14)(25)(36),(23)(56)))$ & 7 & $\{111,122,133,212,264,321,346\}$\\
 &&&8& $\{111,122,133,144,212,231,354,365\}$\\
 &&&9& $\{111,122,133,155,231,245,256,321,332\}$\\ 
 & $c_{6.2.6}$ & $(\mathrm{Id}_3;((123)(456),(14)(25)(36),(153426)))$ & 4 & $\{111,122,414,435\}$\\
 & $c_{6.2.7}$ & $(\mathrm{Id}_3;((456),(456),(456)))$ & 8 & $\{111,144,223,245,414,426,441,453\}$\\
 &&&9& $\{111,122,144,212,245,414,426,441,453\}$\\
 & $c_{6.2.8}$ & $(\mathrm{Id}_3;((123)(456),(123)(465),(132)))$ & 6 & $\{111,144,155,414,426,441\}$\\
&&&7& $\{111,122,144,155,414,435,441\}$\\
& $c_{6.2.9}$ & $(\mathrm{Id}_3;((456),(152634),(152634)))$ & 3 & $\{111,212,426\}$\\
& $c_{6.2.10}$ & $(\mathrm{Id}_3;((153624),(152634),(23)(45)))$ & 4 & $\{111,122,144,155\}$\\   \midrule
$L_{6.3}$ & $c_{6.3.2}$ & $(\mathrm{Id}_3;(\mathrm{Id}_6,(123)(465),(123)(465)))$ & 8 & $\{111,144,212,245,414,443,515,642\}$\\
& $c_{6.3.3}$ & $(\mathrm{Id}_3;((23)(56),(15)(24)(36),(15)(24)(36)))$ & 6 & $\{111,122,212,223,515,562\}$\\
&&&7& $\{111,122,133,212,223,414,541\}$\\
&&&8& $\{111,122,133,212,245,414,426,541\}$\\
& $c_{6.3.4}$ & $(\mathrm{Id}_3;((23)(56),(143526),(143526)))$ & 4 & $\{111,212,414,515\}$\\
& $c_{6.3.5}$ & $(\mathrm{Id}_3;((456),(456),(456)))$ & 8 & $\{111,144,223,245,414,426,443,452\}$\\
&&&9& $\{111,122,144,212,245,414,426,443,452\}$\\
& $c_{6.3.6}$ & $(\mathrm{Id}_3;((123)(456),(123)(465),(132)))$ & 6 & $\{111,144,155,414,426,443\}$\\ \midrule
$L_{6.4}$ & $c_{6.4.2}$ & $(\mathrm{Id}_3;(\mathrm{Id}_6,(16)(25)(34),(16)(25)(34)))$ & 9 & $\{111,122,133,212,223,313,426,435,536\}$\\
&&&10& $\{111,122,133,212,223,313,321,414,435,536\}$\\
 & $c_{6.4.3}$ & $(\mathrm{Id}_3;(\mathrm{Id}_6,(123)(465),(123)(465)))$ & 8 & $\{111,144,212,246,414,442,515,643\}$\\
 & $c_{6.4.4}$ & $(\mathrm{Id}_3;(\mathrm{Id}_6,(142635),(142635)))$ & 5 & $\{111,212,313,414,515\}$\\
 & $c_{6.4.5}$ & $(\mathrm{Id}_3;((23)(45),(23)(45),(23)(45)))$ & 7 & $\{111,122,144,212,246,435,442\}$\\
 &&&8& $\{111,122,144,212,246,414,426,442\}$\\
  &&&9& $\{111,122,144,212,246,414,442,463,625\}$\\
& $c_{6.4.6}$ & $(\mathrm{Id}_3;((23)(45),(14)(25)(36),(14)(25)(36)))$ & 6 & $\{111,122,133,231,435,451\}$\\
 &&&7& $\{111,122,133,212,231,426,634\}$\\
&&&8& $\{111,122,133,212,223,246,435,616\}$\\
& $c_{6.4.7}$ & $(\mathrm{Id}_3;((123)(465),(123)(465),(132)(456)))$ & 4 & $\{111,144,414,442\}$\\
 &&&5& $\{111,122,144,426,451\}$\\
& $c_{6.4.8}$ & $(\mathrm{Id}_3;((123)(465),(142635),(16)(25)(34)))$ & 4 & $\{111,122,414,426\}$\\
& $c_{6.4.9}$ & $(\mathrm{Id}_3;((123)(465),(153624),(142635)))$ & 3 & $\{111,122,426\}$\\
\bottomrule
\end{tabularx}
\end{table}

\clearpage

\renewcommand{\tabcolsep}{5pt}
\begin{table}[htbp]
\caption{Critical sets based on non-trivial isotopisms  (II).}
    \label{Table_scs_lcs_6_app}
\centering
\tiny
\begin{tabularx}{\textwidth}{l l X c X}
\toprule
$L$ & Class & Autotopism & Size & Critical set \\
   \midrule
$L_{6.5}$ &  $c_{6.5.2}$ & $(\mathrm{Id}_3;(\mathrm{Id}_6,(15)(24)(36),(15)(24)(36)))$ & 9 & $\{111,122,133,212,223,335,421,532,616\}$\\
&&&10& $\{111,122,133,212,223,231,313,421,616,634\}$\\
&  $c_{6.5.3}$ & $(\mathrm{Id}_3;((12)(45),(23)(46),(12)(45)))$ & 7 & $\{111,144,313,324,421,463,641\}$\\
&&&8& $\{111,122,144,313,324,342,421,463\}$\\
&&&9& $\{111,122,133,155,313,445,452,625,641\}$\\
&  $c_{6.5.4}$ & $(\mathrm{Id}_3;((12)(45),(15)(26)(34),(14)(25)(36)))$ & 6 & $\{111,122,133,324,421,616\}$\\
&&& 7 & $\{111,122,133,313,324,414,625\}$\\
&&& 8 & $\{111,122,133,144,414,436,616,625\}$\\
&  $c_{6.5.5}$ & $(\mathrm{Id}_3;((345),(123)(465),(123)(465)))$ & 4 & $\{111,246,313,356\}$\\
&&& 5 & $\{111,144,212,324,356\}$\\
&&& 6 & $\{111,144,212,246,324,616\}$\\
&  $c_{6.5.6}$ & $(\mathrm{Id}_3;((345),(143526),(143526)))$ & 3 & $\{111,212,313\}$\\    \midrule
$L_{6.6}$ &  $c_{6.6.2}$ & $(\mathrm{Id}_3;((45),(15)(26)(34),(15)(26)(34)))$ & 7 & $\{111,122,212,414,436,625,634\}$\\
&&& 8 & $\{111,122,133,212,313,414,436,625\}$\\
&&& 9 & $\{111,122,133,212,224,235,313,414,625\}$\\
&  $c_{6.6.3}$ & $(\mathrm{Id}_3;((23)(45),(16)(25),(16)(25)))$ & 7 & $\{111,122,133,212,224,445,451\}$\\
&&& 8 & $\{111,122,133,212,224,235,451,462\}$\\
&&& 9 & $\{111,122,133,212,224,235,263,414,451\}$\\
&  $c_{6.6.4}$ & $(\mathrm{Id}_3;((345),(142)(356),(142)(356)))$ & 4 & $\{111,235,321,365\}$\\
&&& 5 & $\{111,133,212,313,365\}$\\
&&& 6 & $\{111,133,212,235,313,616\}$\\
&  $c_{6.6.5}$ & $(\mathrm{Id}_3;((12)(345),(152346),(162543)))$ & 2 & $\{111,321\}$\\
&&& 3 & $\{111,313,616\}$\\
&  $c_{6.6.6}$ & $(\mathrm{Id}_3;((2345),(1364),(1364)))$ & 5 & $\{111,122,212,224,241\}$\\
&&& 6 & $\{111,122,212,224,235,256\}$\\
&  $c_{6.6.7}$ & $(\mathrm{Id}_3;((12345),(26543),(12345)))$ & 3 & $\{111,122,133\}$\\
&&& 4 & $\{122,133,144,155\}$\\     \midrule
$L_{6.7}$ &  $c_{6.7.2}$ & $(\mathrm{Id}_3;(\mathrm{Id}_6,(15)(26)(34),(15)(26)(34)))$ & 9 & $\{111,122,133,212,224,313,425,436,631\}$\\
&&& 10 & $\{111,122,133,212,224,313,321,414,436,631\}$\\
 &  $c_{6.7.3}$ & $(\mathrm{Id}_3;((23)(46),(16)(25),(16)(25)))$ & 7 & $\{111,122,133,212,256,414,442\}$\\
 &&& 8 & $\{111,122,133,212,224,241,414,442\}$\\
  &&& 9 & $\{111,122,133,212,224,235,414,453,526\}$\\
 &  $c_{6.7.4}$ & $(\mathrm{Id}_3;((23)(46),(12)(34)(56),(12)(34)(56)))$ & 6 & $\{111,133,212,235,263,425\}$\\
  &&& 7 & $\{111,133,155,224,235,425,461\}$\\
  &&& 8 & $\{111,133,155,224,235,414,436,461\}$\\ 
&  $c_{6.7.5}$ & $(\mathrm{Id}_3;((123)(456),(164)(235),(136)(254)))$ & 4 & $\{111,122,133,414\}$\\
  &&& 5 & $\{111,133,144,166,461\}$\\
&  $c_{6.7.6}$ & $(\mathrm{Id}_3;((123)(456),(124563),(146532)))$ & 3 & $\{111,122,425\}$\\     \midrule
$L_{6.8}$ &  $c_{6.8.2}$ & $(\mathrm{Id}_3;((24),(14)(23)(56),(14)(23)(56)))$ & 7 & $\{111,122,212,256,325,616,651\}$\\
  &&& 8 & $\{111,122,155,212,235,313,515,651\}$\\
  &&& 9 & $\{111,122,155,212,235,313,325,616,651\}$\\  
  &&& 10 & $\{111,122,155,224,263,313,325,354,515,521\}$\\
 &  $c_{6.8.3}$ & $(\mathrm{Id}_3;((12)(34),(14)(36),(24)(56)))$ & 7 & $\{111,155,336,342,515,521,534\}$\\
  &&& 8 & $\{111,122,133,313,336,515,521,632\}$\\
  &&& 9 & $\{111,122,133,144,313,336,515,521,534\}$\\
 &  $c_{6.8.4}$ & $(\mathrm{Id}_3;((1234),(2653),(1234)))$ & 5 & $\{111,122,166,515,521\}$\\
  &&& 6 & $\{111,122,133,155,515,521\}$\\  \midrule
$L_{6.9}$ &  $c_{6.9.2}$ & $(\mathrm{Id}_3;((12)(34),(16)(25)(34),(13)(26)(45)))$ & 6 & $\{111,122,133,313,334,532\}$\\
  &&& 7 & $\{111,122,133,144,326,334,515\}$\\
&&& 8 & $\{111,122,133,144,326,515,532,631\}$\\    \midrule
$L_{6.10}$ &  $c_{6.10.2}$ & $(\mathrm{Id}_3;((12),(14)(26)(35),(13)(24)(56)))$ & 7 & $\{111,122,133,326,423,515,634\}$\\
&&& 8 & $\{111,122,133,155,332,423,531,616\}$\\
&&& 9 & $\{111,122,133,313,326,524,616,621,634\}$\\
 &  $c_{6.10.3}$ & $(\mathrm{Id}_3;((12)(34),(25)(36),(12)(34)))$ & 7 & $\{111,122,133,326,332,515,524\}$\\
 &&& 8 & $\{111,122,133,166,332,354,515,524\}$\\
 &&& 9 & $\{111,122,133,326,365,515,524,531,621\}$\\  \midrule
$L_{6.11}$ & $c_{6.11.2}$  & $(\mathrm{Id}_3;((12)(35),(12)(56),(16)(34)))$ & 7 & $\{111,133,155,313,341,435,654\}$\\
 &&& 8 & $\{111,122,133,155,336,364,414,435\}$\\
 &&& 9 & $\{111,122,133,155,325,336,364,435,616\}$\\  \midrule
$L_{6.12}$ & $c_{6.12.2}$ & $(\mathrm{Id}_3;(\mathrm{Id}_6,(126)(345),(126)(345)))$ & 8 & $\{111,133,212,234,313,332,414,536\}$\\
 & $c_{6.12.3}$ & $(\mathrm{Id}_3;((12),(13)(25)(46),(14)(23)(56)))$ & 7 & $\{111,122,144,324,425,515,643\}$\\
 &&& 8 & $\{111,122,133,324,346,425,616,643\}$\\
 &&& 9 & $\{111,122,133,144,313,425,616,621,643\}$\\
 & $c_{6.12.4}$ & $(\mathrm{Id}_3;((45),(132465),(132465)))$ & 4 & $\{111,212,313,414\}$\\
  &&& 5 & $\{111,212,414,431,616\}$\\
 & $c_{6.12.5}$ & $(\mathrm{Id}_3;((12)(45),(26)(35),(12)(45)))$ & 7 & $\{111,122,133,313,425,431,635\}$\\
 &&& 8 & $\{111,122,133,155,313,324,425,431\}$\\
&&& 9 & $\{111,122,133,144,313,324,425,442,456\}$\\
 & $c_{6.12.6}$ & $(\mathrm{Id}_3;((345),(345),(345)))$ & 8 & $\{111,133,226,234,313,324,332,346\}$\\
&&& 9 & $\{111,122,133,212,234,313,324,332,346\}$\\
 & $c_{6.12.7}$ & $(\mathrm{Id}_3;((345),(126)(354),(126)(354)))$ & 6 & $\{111,133,212,313,332,635\}$\\
 &&& 7 & $\{111,133,212,234,313,324,332\}$\\
 & $c_{6.12.8}$ & $(\mathrm{Id}_3;((12)(345),(136425),(146523)))$ & 3 & $\{111,133,313\}$\\
\bottomrule
\end{tabularx}
\end{table}

\clearpage

\renewcommand{\tabcolsep}{1pt}
\begin{table}[htbp]
\caption{Critical sets based on non-trivial paratopisms of type I (I).}
    \label{table1_app}
\centering
\tiny
\begin{tabularx}{\textwidth}{l l X c X}
\toprule
$L$ & Class & Autoparatopism & Size & Critical set \\
\midrule
$L_3$ &  $c_{3.5}$ & $((23);(\mathrm{Id}_3,(23),(23)))$ & 2 & $\{111,212\}$\\
&  $c_{3.6}$ & $((23);((23),\mathrm{Id}_3,\mathrm{Id}_3))$ & 1 & $\{212\}$\\
&  $c_{3.7}$ & $((23);((23),(123),(123)))$ & 1 & $\{212\}$\\
&  $c_{3.8}$ & $((23);((123),(23),(13)))$ & 1 & $\{111\}$\\  \midrule
$L_{4.1}$ &  $c_{4.1.7}$ & $((13);(\mathrm{Id}_4,\mathrm{Id}_4,\mathrm{Id}_4))$ & 4 & $\{111,122,234,313\}$\\
 &&& 5 & $\{111,122,133,212,313\}$\\
 &&& 6 & $\{122,133,144,234,243,324\}$\\
&  $c_{4.1.8}$ & $((13);((34),(34),(34)))$ & 3 & $\{111,133,313\}$\\
 &&& 4 & $\{111,133,234,423\}$\\
&  $c_{4.1.9}$ & $((13);((12),(12),(34)))$ & 3 & $\{111,133,313\}$\\
&  $c_{4.1.10}$ & $((13);((12)(34),(12)(34),\mathrm{Id}_4))$ & 3 & $\{111,133,313\}$\\
&  $c_{4.1.11}$ & $((13);((1324),(1324),(34)))$ & 2 & $\{111,122\}$\\
&  $c_{4.1.12}$ & $((13);((243),(243),(243)))$ & 1 & $\{122\}$\\ \midrule
$L_{4.2}$ & $c_{4.2.7}$ & $((23);(\mathrm{Id}_4,(34),(34)))$ & 3 & $\{111,234,313\}$\\ 
 &&& 4 & $\{111,133,234,414\}$\\
  & $c_{4.2.8}$ & $((23);((34),\mathrm{Id}_4,\mathrm{Id}_4))$ & 3 & $\{111,234,313\}$\\ 
 &&& 4 & $\{111,133,313,324\}$\\
  & $c_{4.2.9}$ & $((23);((34),(1324),(1324)))$ & 3 & $\{111,313,332\}$\\ 
  & $c_{4.2.10}$ & $((23);((12)(34),(34),(12)))$ & 3 & $\{111,133,313\}$\\ 
  & $c_{4.2.11}$ & $((23);((13)(24),\mathrm{Id}_4,(1423)))$ & 2 & $\{111,212\}$\\ 
  & $c_{4.2.12}$ & $((23);((13)(24),(34),(14)(23)))$ & 1 & $\{111\}$\\ \midrule
$L_{5.1}$ & $c_{5.1.6}$ & $((13);((25)(34),\mathrm{Id}_5,(25)(34)))$ & 4 & $\{111,122,234,245\}$\\
 &&& 5 & $\{111,122,133,212,352\}$\\
 &&& 6 & $\{111,122,133,144,313,324\}$\\
& $c_{5.1.7}$ & $((13);(\mathrm{Id}_5,(25)(34),\mathrm{Id}_5))$ & 3 & $\{122,133,245\}$\\
 &&& 4 & $\{122,133,144,234\}$\\
 & $c_{5.1.8}$ & $((13);((2453),(2354),(2453)))$ & 3 & $\{122,133,245\}$\\
  & $c_{5.1.9}$ & $((13);((12345),(12)(35),\mathrm{Id}_5))$ & 2 & $\{111,133\}$\\
  & $c_{5.1.10}$ & $((13);((12)(35),(12345),(25)(34)))$ & 1 & $\{111\}$\\ \midrule
$L_{5.2}$ & $c_{5.2.4}$ & $((23);((45),(12)(35),(12)(35)))$ & 4 & $\{111,234,313,425\}$\\
 &&& 5 & $\{111,133,212,324,432\}$\\
 &&& 6 & $\{111,133,212,324,335,451\}$\\
 & $c_{5.2.5}$ & $((23);((1345),(12354),(152)))$ & 3 & $\{111,122,155\}$\\
\midrule
$L_{6.1}$ & $c_{6.1.5}$ & $((23);((45),(12)(36)(45),(12)(36)(45)))$ & 6 & $\{111,234,313,426,435,645\}$\\
&&& 7 & $\{111,133,212,234,313,414,645\}$\\
&&& 8 & $\{111,133,212,221,234,313,426,645\}$\\
&&& 9 & $\{111,133,212,234,313,336,354,414,663\}$\\
& $c_{6.1.6}$ & $((23);((45),(34)(56),(34)(56)))$ & 7 & $\{111,133,256,313,414,426,645\}$\\
&&& 8 & $\{111,122,234,265,313,426,435,645\}$\\
&&& 9 & $\{111,122,133,234,265,313,414,435,453\}$\\
&&& 10 & $\{111,133,155,234,265,336,414,435,453,461\}$\\
& $c_{6.1.7}$ & $((23);((12)(3465),(35)(46),(12)(36)(45)))$ & 4 & $\{111,133,313,336\}$\\
& $c_{6.1.8}$ & $((23);((12)(3465),(12),(34)(56)))$ & 5 & $\{111,133,155,313,336\}$\\  \midrule
$L_{6.2}$ & $c_{6.2.11}$ & $((13);((23)(56),(56),(23)(56)))$ & 7 & $\{111,122,155,245,354,515,536\}$\\
&&& 8 & $\{111,122,133,155,212,264,536,625\}$\\
&&& 9 & $\{111,122,133,155,166,212,256,536,625\}$\\
&&& 10 & $\{111,122,133,155,166,212,223,256,515,536\}$\\
& $c_{6.2.12}$ & $((13);((14)(25)(36),(56),(14)(25)(36)))$ & 8 & $\{111,122,133,212,245,256,346,453\}$\\
&&& 9 & $\{111,122,133,155,212,223,231,346,453\}$\\
&&& 10 & $\{111,122,133,155,212,223,332,346,441,563\}$\\
&&& 11 & $\{111,122,133,155,212,223,231,313,321,332,453\}$\\
& $c_{6.2.13}$ & $((13);((14)(26)(35),(14)(2536),(23)(56)))$ & 4 & $\{111,122,166,212\}$\\
& $c_{6.2.14}$ & $((13);(\mathrm{Id}_6,(14)(2536),(14)(25)(36)))$ & 5 & $\{111,122,166,212,313\}$\\
& $c_{6.2.15}$ & $((13);((143625),(123)(46),(153426)))$ & 4 & $\{111,122,144,441\}$\\
& $c_{6.2.16}$ & $((13);((12)(46),(123)(46),(23)(56)))$ & 4 & $\{111,144,155,414\}$\\
& $c_{6.2.17}$ & $((13);((16)(24)(35),(123)(46),(14)(25)(36)))$ & 4 & $\{111,144,155,441\}$\\
&&&  5 & $\{111,122,133,144,441\}$\\
& $c_{6.2.18}$ & $((13);((153426),(56),(153426)))$ & 5 & $\{111,122,133,155,453\}$\\
& $c_{6.2.19}$ & $((13);((123)(456),(14)(2536),(153426)))$ & 2 & $\{111,122\}$\\ \midrule
$L_{6.3}$ & $c_{6.3.7}$ & $((13);((23)(56),(46),(23)(56)))$ & 6 & $\{111,122,245,354,426,435\}$\\
&&& 7 & $\{111,122,133,144,365,426,435\}$\\
&&& 8 & $\{111,122,133,144,212,256,414,426\}$\\
&&& 9 & $\{111,122,133,144,166,212,264,414,426\}$\\
& $c_{6.3.8}$ & $((13);((152634),(14)(2635),(123)(456)))$ & 4 & $\{111,122,133,264\}$\\
&&& 5 & $\{111,122,133,212,332\}$\\
& $c_{6.3.9}$ & $((13);((12)(46),(123)(45),(23)(56)))$ & 4 & $\{111,144,155,414\}$\\
& $c_{6.3.10}$ & $((13);((143625),(14)(2635),\mathrm{Id}_6))$ & 2 & $\{111,122\}$\\
\bottomrule
\end{tabularx}
\end{table}

\renewcommand{\tabcolsep}{1pt}
\begin{table}[htbp]
\caption{Critical sets based on non-trivial paratopisms of type I (II).}
    \label{table1_app_1}
\centering
\tiny
\begin{tabularx}{\textwidth}{l l X c X}
\toprule
$L$ & Class & Autoparatopism & Size & Critical set \\ \midrule
$L_{6.4}$ & $c_{6.4.10}$ & $((23);(\mathrm{Id}_6,(23)(45),(23)(45)))$ & 6 & $\{111,144,223,345,426,515\}$\\
&&& 7 & $\{111,122,144,212,356,435,524\}$\\
&&& 8 & $\{111,122,144,212,254,313,435,524\}$\\
&&& 9 & $\{111,122,144,212,254,313,414,435,515\}$\\
&&& 10 & $\{111,166,223,254,414,435,515,524,616,625\}$\\
& $c_{6.4.11}$ & $((23);((23)(45),\mathrm{Id}_6,\mathrm{Id}_6))$ & 6 & $\{111,212,246,435,625,634\}$\\
&&& 7 & $\{111,122,212,223,246,426,634\}$\\
&&& 8 & $\{111,122,133,212,246,414,435,451\}$\\
&&& 9 & $\{111,122,133,212,223,231,246,414,463\}$\\
& $c_{6.4.12}$ & $((23);((14)(25)(36),(123)(465),(142635)))$ & 6 & $\{111,122,212,223,231,321\}$\\
& $c_{6.4.13}$ & $((23);((23)(45),(123)(465),(123)(465)))$ & 4 & $\{212,246,414,616\}$\\
&&& 5 & $\{111,212,246,414,442\}$\\
& $c_{6.4.14}$ & $((23);((123)(465),(23)(45),(13)(46)))$ & 3 & $\{111,144,414\}$\\
& $c_{6.4.15}$ & $((23);((14)(25)(36),\mathrm{Id}_6,(153624)))$ & 3 & $\{111,212,313\}$\\
& $c_{6.4.16}$ & $((23);((142635),(23)(45),(15)(26)(34)))$ & 2 & $\{111,122\}$\\ \midrule
$L_{6.5}$ & $c_{6.5.7}$ & $((23);((45),(13)(56),(13)(56)))$ & 7 & $\{111,122,246,335,421,445,616\}$\\
&&& 8 & $\{111,122,144,212,335,414,421,616\}$\\
&&& 9 & $\{111,122,144,212,231,313,414,436,625\}$\\
&&& 10 & $\{111,122,155,231,335,356,414,436,445,463\}$\\
 & $c_{6.5.8}$ & $((23);((45),(16)(24)(35),(16)(24)(35)))$ & 6 & $\{111,231,342,452,616,625\}$\\
 &&& 7 & $\{111,122,212,313,324,361,445\}$\\
 &&& 8 & $\{111,122,212,223,324,361,421,653\}$\\
 &&& 9 & $\{111,122,212,223,313,361,414,452,625\}$\\
 &&& 10 & $\{111,122,313,324,335,342,414,452,625,653\}$\\
 & $c_{6.5.9}$ & $((23);((12)(345),\mathrm{Id}_6,(132)(456)))$ & 4 & $\{111,144,313,324\}$\\
 & $c_{6.5.10}$ & $((23);((12)(345),(15)(24)(36),(162534)))$ & 2 & $\{111,313\}$\\
 &&& 4 & $\{111,144,324,616\}$\\ \midrule
$L_{6.6}$ & $c_{6.6.8}$ & $((23);(\mathrm{Id}_6,(16)(25)(34),(16)(25)(34)))$ & 8 & $\{111,122,212,313,451,625,634,643\}$\\
 &&& 9 & $\{111,122,212,224,321,451,515,634,643\}$\\
 &&& 10 & $\{111,122,212,224,321,515,616,643,652,661\}$\\
 &&& 11 & $\{111,122,212,224,321,332,414,451,515,531,643\}$\\
 &&& 15 & $\{111,122,133,212,224,241,313,321,332,414,423,$\\
  &&&  &  $\phantom{\{}451,515,531,542\}$\\
& $c_{6.6.9}$ & $((23);((45),(12)(56),(12)(56)))$ & 7 & $\{111,133,212,346,414,451,625\}$\\
&&& 8 & $\{111,133,144,212,346,414,423,625\}$\\
&&& 9 & $\{111,133,144,212,224,321,423,445,616\}$\\
&&& 10 & $\{111,133,144,212,224,313,321,414,445,451\}$\\
&&& 11 & $\{111,133,144,212,224,313,321,414,423,436,451\}$\\
& $c_{6.6.10}$ & $((23);((23)(45),(34),(34)))$ & 6 & $\{111,133,212,224,445,625\}$\\
&&& 7 & $\{111,122,133,212,224,235,451\}$\\
&&& 8 & $\{111,122,155,212,224,235,451,643\}$\\
&&& 9 & $\{111,122,155,224,235,256,462,634,643\}$\\
& $c_{6.6.11}$ & $((23);((2345),(1463)(25),(1463)(25)))$ & 5 & $\{111,122,212,224,241\}$\\
&&& 6 & $\{111,122,212,224,235,256\}$\\
& $c_{6.6.12}$ & $((23);((345),(132645),(132645)))$ & 3 & $\{111,212,321\}$\\
&&& 4 & $\{111,212,313,616\}$\\
& $c_{6.6.13}$ & $((23);((12)(345),(124),(365)))$ & 4 & $\{111,133,313,321\}$\\
& $c_{6.6.14}$ & $((23);((12345),(162)(35),(156)(24)))$ & 1 & $\{111\}$\\  \midrule
$L_{6.7}$ & $c_{6.7.7}$ & $((23);(\mathrm{Id}_6,(12)(56),(12)(56)))$ & 7 & $\{111,133,144,212,346,425,623\}$\\
&&& 8 & $\{111,133,144,212,224,313,425,645\}$\\
&&& 9 & $\{111,133,144,212,224,256,321,414,534\}$\\
&&& 10 & $\{111,133,144,212,224,313,321,414,425,623\}$\\
 & $c_{6.7.8}$ & $((23);((23)(46),(34),(34)))$ & 6 & $\{111,133,212,414,442,526\}$\\
 &&& 7 & $\{111,122,133,212,224,414,425\}$\\
 &&& 8 & $\{111,122,133,212,224,235,425,461\}$\\
 &&& 9 & $\{111,122,133,224,263,414,436,442,461\}$\\
 & $c_{6.7.9}$ & $((23);((14)(26)(35),(123)(456),(142536)))$ & 6 & $\{111,122,212,224,235,321\}$\\
  & $c_{6.7.10}$ & $((23);((123)(456),(2463),(1453)))$ & 3 & $\{111,122,425\}$\\
& $c_{6.7.11}$ & $((23);((142536),(1354),(15)(2463)))$ & 2 & $\{111,133\}$\\ \midrule
$L_{6.8}$ & $c_{6.8.5}$ & $((23);(\mathrm{Id}_6,(15)(23)(46),(15)(23)(46)))$ & 7 & $\{111,122,212,431,515,623,664\}$\\
 &&& 8 & $\{111,122,212,224,336,515,651,664\}$\\
  &&& 9 & $\{111,122,212,224,313,414,534,546,651\}$\\
   &&& 10 & $\{111,122,212,224,313,414,616,623,632,664\}$\\
    &&& 11 & $\{111,122,212,224,313,354,414,431,515,521,534\}$\\
 & $c_{6.8.6}$ & $((23);((24),(16)(45),(16)(45)))$ & 7 & $\{111,122,212,235,313,645,651\}$\\
  &&& 8 & $\{111,122,133,212,235,313,515,651\}$\\
   &&& 9 & $\{111,122,133,212,235,256,361,616,623\}$\\
 & $c_{6.8.7}$ & $((23);((13)(24),(153462),(126435)))$ & 6 & $\{111,122,256,515,623,632\}$\\
   &&& 7 & $\{111,122,133,212,515,521,632\}$\\
     &&& 8 & $\{111,122,133,155,224,515,534,623\}$\\
       &&& 9 & $\{111,122,133,155,212,224,515,521,534\}$\\
& $c_{6.8.8}$ & $((23);((1234),(15246),(13645)))$ & 5 & $\{111,122,133,515,521\}$\\
   &&& 6 & $\{111,144,155,166,515,521\}$\\
\bottomrule
\end{tabularx}
\end{table}

\clearpage

\renewcommand{\tabcolsep}{2pt}
\begin{table}[htbp]
\caption{Critical sets based on non-trivial paratopisms of type I (III).}
    \label{table1_app_2}
\centering
\tiny
\begin{tabularx}{\textwidth}{l l X c X}
\toprule
$L$ & Class & Autoparatopism & Size & Critical set \\  \midrule
$L_{6.9}$ & $c_{6.9.3}$ & $((23);((34),(12)(35),(12)(35)))$ & 6 & $\{111,133,235,326,546,625\}$\\
&&& 7 & $\{111,133,144,224,313,515,625\}$\\
&&& 8 & $\{111,133,144,212,224,313,334,625\}$\\
&&& 9 & $\{111,133,144,212,224,235,326,342,625\}$\\
&&& 10 & $\{111,133,166,212,235,256,313,351,365,616\}$\\
 & $c_{6.9.4}$ & $((23);((1324)(56),(13526),(3654)))$ & 4 & $\{111,122,133,532\}$\\
 &&& 5 & $\{111,122,133,144,515\}$\\ \midrule
$L_{6.10}$ & $c_{6.10.4}$ & $((23);(\mathrm{Id}_6,(16)(23)(45),(16)(23)(45)))$ & 6 & $\{122,332,414,515,621,634\}$\\
&&& 7 & $\{111,122,212,332,423,524,616\}$\\
&&& 8 & $\{111,122,212,225,313,461,515,634\}$\\
&&& 9 & $\{111,122,212,225,313,332,423,616,621\}$\\
 & $c_{6.10.5}$ & $((23);((34),(15)(46),(15)(46)))$ & 6 & $\{111,122,236,332,524,616\}$\\
&&& 7 & $\{111,122,133,212,332,531,616\}$\\
&&& 8 & $\{111,122,133,212,251,332,515,634\}$\\
&&& 9 & $\{111,122,133,212,236,313,332,515,531\}$\\
&&& 10 & $\{111,122,144,326,354,365,515,524,546,616\}$\\
 & $c_{6.10.6}$ & $((23);((12),(152)(346),(125)(364)))$ & 7 & $\{111,133,313,332,553,616,621\}$\\
 &&& 8 & $\{111,122,133,155,332,515,634,663\}$\\
 &&& 9 & $\{111,122,133,155,313,423,515,553,634\}$\\
 &&& 10 & $\{111,122,144,326,515,524,546,616,621,652\}$\\
  & $c_{6.10.7}$ & $((23);((12)(34),(162453),(135426)))$ & 7 & $\{111,122,313,341,515,524,652\}$\\
 &&& 8 & $\{111,122,133,313,332,524,553,616\}$\\
 &&& 9 & $\{111,122,133,155,332,354,524,553,616\}$\\
 &&& 10 & $\{111,122,166,326,332,341,365,531,553,645\}$\\ \midrule
$L_{6.11}$ & $c_{6.11.3}$ & $((23);((23),(16)(23)(45),(16)(23)(45)))$ & 6 & $\{111,212,234,414,621,645\}$\\
 &&& 7 & $\{111,122,212,243,414,621,645\}$\\
 &&& 8 & $\{111,122,212,226,243,423,645,654\}$\\
 &&& 9 & $\{111,122,212,226,234,414,461,645,654\}$\\
 &&& 10 & $\{111,122,212,226,414,423,616,621,632,654\}$\\
 & $c_{6.11.4}$ & $((23);((1253),(2453),(154623)))$ & 5 & $\{111,122,144,414,435\}$\\
  &&& 6 & $\{122,133,144,155,414,435\}$\\ \midrule
$L_{6.12}$ & $c_{6.12.9}$ & $((23);(\mathrm{Id}_6,(13)(25)(46),(13)(25)(46)))$ & 7 & $\{111,122,313,346,414,552,621\}$\\
 &&& 8 & $\{111,122,212,313,332,425,536,654\}$\\
 &&& 9 & $\{111,122,212,313,324,414,431,515,552\}$\\
 &&& 10 & $\{111,122,212,313,324,332,414,425,515,552\}$\\
  & $c_{6.12.10}$ & $((23);((45),(26)(35),(26)(35)))$ & 6 & $\{111,133,313,442,643,662\}$\\
  &&& 7 & $\{111,122,133,212,332,456,643\}$\\
 &&& 8 & $\{111,122,133,212,253,313,332,456\}$\\
 &&& 9 & $\{111,122,133,212,253,313,324,456,662\}$\\
 &&& 10 & $\{111,122,133,313,332,431,456,463,635,662\}$\\
 & $c_{6.12.11}$ & $((23);((12),(162)(354),(126)(345)))$ & 7 & $\{111,133,324,442,616,621,662\}$\\
  &&& 8 & $\{111,122,133,313,425,616,621,662\}$\\
 &&& 9 & $\{111,122,133,313,425,616,643,654,662\}$\\
 &&& 11 & $\{111,133,313,324,332,414,425,442,515,523,552\}$\\
 & $c_{6.12.12}$ & $((23);((12)(45),(132465),(156423)))$ & 7 & $\{111,122,313,324,414,456,635\}$\\
  &&& 8 & $\{111,122,313,324,414,425,442,635\}$\\
 &&& 9 & $\{111,122,133,313,324,332,414,425,442\}$\\
 &&& 10 & $\{111,122,133,313,324,332,414,425,431,463\}$\\
  & $c_{6.12.13}$ & $((23);((12)(345),(126)(354),(345)))$ & 4 & $\{111,133,313,324\}$\\
  & $c_{6.12.14}$ & $((23);((12),\mathrm{Id}_6,(162)(354)))$ & 5 & $\{111,133,313,414,515\}$\\
& $c_{6.12.15}$ & $((23);((345),(136425),(136425)))$ & 4 & $\{111,212,313,332\}$\\
& $c_{6.12.16}$ & $((23);((12)(45),(14)(25)(36),(132465)))$ & 4 & $\{111,414,431,616\}$\\
&&& 5 & $\{111,133,313,414,431\}$\\
& $c_{6.12.17}$ & $((23);((12)(345),(345),(162)))$ & 4 & $\{111,133,313,616\}$\\
&&& 5 & $\{111,133,313,324,332\}$\\
\bottomrule
\end{tabularx}
\end{table}

\clearpage

\renewcommand{\tabcolsep}{8pt}
\begin{table}[htbp]
\caption{Critical sets based on non-trivial paratopisms of type II.}
    \label{table2_app}
\centering
\tiny
\begin{tabularx}{\textwidth}{l l X c X}
\toprule
$L$ & Class & Autoparatopism & Size & Critical set \\
\midrule
$L_3$ &  $c_{3.9}$ & $((123);(\mathrm{Id}_3,(23),(23)))$ & 1 & $\{111\}$\\
 &  $c_{3.10}$ & $((123);((23),\mathrm{Id}_3,\mathrm{Id}_3))$ & 0 & $\{\}$\\ \midrule
$L_{4.1}$ &  $c_{4.1.13}$ & $((132);(\mathrm{Id}_4,\mathrm{Id}_4,\mathrm{Id}_4))$ & 3 & $\{111,122,133\}$\\
 &  $c_{4.1.14}$ & $((132);((243),(243),(243)))$ & 2 & $\{122,133\}$\\
 &&& 3 & $\{111,122,234\}$\\
 &  $c_{4.1.15}$ & $((132);((34),(34),(34)))$ & 2 & $\{111,133\}$\\
 &  $c_{4.1.16}$ & $((132);((123),(243),(123)))$ & 1 & $\{111\}$\\\midrule
$L_{4.2}$ & $c_{4.2.13}$ & $((123);(\mathrm{Id}_4,(34),(34)))$ & 2 & $\{111,234\}$\\
 &&& 3 & $\{111,133,144\}$\\
  & $c_{4.2.14}$ & $((123);((34),\mathrm{Id}_4,\mathrm{Id}_4))$ & 2 & $\{111,133\}$\\ \midrule
$L_{5.1}$ & $c_{5.1.11}$ & $((132);((25)(34),\mathrm{Id}_5,(25)(34)))$ & 2 & $\{223,352\}$\\
 &&& 3 & $\{111,223,234\}$\\
 &&& 4 & $\{111,122,155,223\}$\\
 & $c_{5.1.12}$ & $((132);(\mathrm{Id}_5,(25)(34),\mathrm{Id}_5))$ & 1 & $\{223\}$\\
 & $c_{5.1.13}$ & $((132);((2453),(2354),(2453)))$ & 1 & $\{223\}$\\ \midrule
$L_{5.2}$ & $c_{5.2.6}$ & $((123);((12)(345),(12),(354)))$ & 3 & $\{111,133,245\}$\\
 &&& 4 & $\{111,133,144,234\}$\\
 &&& 5 & $\{212,234,245,253,335\}$\\
& $c_{5.2.7}$ & $((123);((12),(12)(354),(345)))$ & 3 & $\{111,133,245\}$\\
 &&& 4 & $\{111,133,144,234\}$\\
 &&& 5 & $\{133,144,155,212,335\}$\\
 &&& 6 & $\{133,144,212,335,443,554\}$\\
& $c_{5.2.8}$ & $((123);((12)(354),(12)(345),\mathrm{Id}_5))$ & 4 & $\{111,133,212,234\}$\\
 &&& 5 & $\{111,133,234,245,253\}$\\
& $c_{5.2.9}$ & $((123);((1342),(1243),(14)(35)))$ & 1 & $\{133\}$\\ \midrule
$L_{6.1}$ & $c_{6.1.9}$ & $((123);(\mathrm{Id}_6,(12)(36),(12)(36)))$ & 4 & $\{111,133,256,354\}$\\
 &&& 5 & $\{111,133,144,221,256\}$\\
 &&& 6 & $\{111,133,144,155,234,243\}$\\
 &&& 7 & $\{133,144,155,221,256,336,663\}$\\
& $c_{6.1.10}$ & $((123);(\mathrm{Id}_6,(3564),(3564)))$ & 3 & $\{111,133,234\}$\\
 &&& 4 & $\{111,133,144,243\}$\\ \midrule
$L_{6.2}$ & $c_{6.2.20}$ & $((132);((23),\mathrm{Id}_6,(23)))$ & 5 & $\{111,122,133,144,256\}$\\
 &&& 6 & $\{111,122,133,144,155,245\}$\\
  &&& 7 & $\{111,245,256,264,346,354,365\}$\\ 
  & $c_{6.2.21}$ & $((132);((14)(2635),\mathrm{Id}_6,(14)(2635)))$ & 3 & $\{111,122,223\}$\\
 &&& 4 & $\{111,122,155,453\}$\\ 
  & $c_{6.2.22}$ & $((132);((12)(465),\mathrm{Id}_6,(12)(465)))$ & 3 & $\{111,144,245\}$\\ \midrule
$L_{6.3}$ & $c_{6.3.11}$ & $((132);((13)(456),(456),(13)(456)))$ & 4 & $\{111,122,245,354\}$\\
 &&& 5 & $\{111,122,144,231,256\}$\\
  &&& 6 & $\{111,122,144,155,231,354\}$\\
   &&& 7 & $\{122,144,155,166,245,256,264\}$\\
 & $c_{6.3.12}$ & $((132);((12)(465),(456),(12)(465)))$ & 3 & $\{111,144,245\}$\\ \midrule
$L_{6.4}$ & $c_{6.4.17}$ & $((123);(\mathrm{Id}_6,(12)(56),(12)(56)))$ & 3 & $\{111,144,356\}$\\
 &&& 5 & $\{111,133,144,246,254\}$\\
  & $c_{6.4.18}$ & $((123);(\mathrm{Id}_6,(23)(45),(23)(45)))$ & 5 & $\{111,122,133,254,345\}$\\
 &&& 6 & $\{111,122,144,223,265,332\}$\\
  & $c_{6.4.19}$ & $((123);((23)(45),\mathrm{Id}_6,\mathrm{Id}_6))$ & 3 & $\{111,122,254\}$\\
 &&& 4 & $\{111,144,223,246\}$\\ \midrule
$L_{6.7}$ & $c_{6.7.12}$ & $((123);((26)(34),(2463),(2463)))$ & 5 & $\{122,133,155,224,332\}$\\
 &&& 6 & $\{111,122,133,155,263,346\}$\\
  &&& 7 & $\{111,122,133,144,166,263,346\}$\\
 & $c_{6.7.13}$ & $((123);((24)(36),(132546),(132546)))$ & 3 & $\{111,133,321\}$\\
 &&& 4 & $\{111,122,321,534\}$\\ \midrule
$L_{6.9}$ & $c_{6.9.5}$ & $((123);((25364),(23456),(23456)))$ & 4 & $\{111,155,334,445\}$\\
 &&& 5 & $\{111,122,166,224,445\}$\\
  &&& 6 & $\{111,122,133,235,256,423\}$\\
   &&& 7 & $\{111,122,133,144,256,334,423\}$\\
    &&& 8 & $\{111,122,133,224,256,263,334,423\}$\\  \midrule
$L_{6.11}$ & $c_{6.11.5}$ & $((123);((26435),(24563),(24563)))$ & 4 & $\{144,243,336,645\}$\\
 &&& 5 & $\{111,122,226,243,435\}$\\
  &&& 6 & $\{111,122,133,144,226,645\}$\\
   &&& 7 & $\{111,122,133,144,234,524,645\}$\\
\bottomrule
\end{tabularx}
\end{table}

\clearpage

\end{appendices}

\newpage

\bibliographystyle{unsrt}
\bibliography{ref_arxiv}

\end{document}